\documentclass[letter,11pt]{article}
\usepackage[margin=1in]{geometry}
\usepackage{mathptmx}
\usepackage[utf8]{inputenc}
\usepackage[T1]{fontenc}
\usepackage[english]{babel}
\usepackage{csquotes}
\usepackage{amsmath}
\usepackage{amsthm}
\usepackage{amsfonts}
\usepackage{commath}
\usepackage{bm}
\usepackage{xfrac}
\usepackage{mathtools}
\usepackage{authblk}
\usepackage{url}
\usepackage[pdfencoding=auto,psdextra,hidelinks]{hyperref}
\usepackage{bookmark}
\usepackage{float}
\usepackage{booktabs}
\usepackage{subcaption}
\usepackage{tikz}
\usetikzlibrary{calc,positioning,shapes,arrows}
\usepackage[shortlabels]{enumitem}

\newcommand*{\email}[1]{\href{mailto:#1}{\nolinkurl{#1}}}

\DeclareMathOperator{\sgn}{sgn}

\newcommand{\Simplex}{\ensuremath{\Delta}}

\newcommand{\eps}{\ensuremath{\varepsilon}}
\newcommand{\NN}{\ensuremath{\mathbb{N}}}

\newcommand{\QQ}{\ensuremath{\mathbb{Q}}}
\newcommand{\RR}{\ensuremath{\mathbb{R}}}

\newcommand{\zo}{\ensuremath{\{0,1\}}}

\newcommand{\PTIME}{\ensuremath{\mathrm{P}}}
\newcommand{\NP}{\ensuremath{\mathrm{NP}}}

\newcommand{\PSPACE}{\ensuremath{\mathrm{PSPACE}}}

\newcommand{\cETR}{\ensuremath{\exists\RR}}

\newcommand{\FP}{\ensuremath{\mathrm{FP}}}
\newcommand{\FNP}{\ensuremath{\mathrm{FNP}}}
\newcommand{\TFNP}{\ensuremath{\mathrm{TFNP}}}
\newcommand{\PPA}{\ensuremath{\mathrm{PPA}}}
\newcommand{\PPAD}{\ensuremath{\mathrm{PPAD}}}
\newcommand{\LinearFIXP}{\ensuremath{\mathrm{LinearFIXP}}}
\newcommand{\FIXP}{\ensuremath{\mathrm{FIXP}}}
\newcommand{\FIXPA}{\ensuremath{\mathrm{FIXP}_a}}
\newcommand{\LinearBU}{\ensuremath{\mathrm{LinearBU}}}
\newcommand{\BU}{\ensuremath{\mathrm{BU}}}
\newcommand{\BUA}{\ensuremath{\mathrm{BU}_a}}
\newcommand{\BBU}{\ensuremath{\mathrm{BBU}}}
\newcommand{\BBUA}{\ensuremath{\mathrm{BBU}_a}}

\newcommand{\CH}{\ensuremath{\mathrm{CH}}}
\newcommand{\CHA}{\ensuremath{\mathrm{CH}_a}}
\newcommand{\BCH}{\ensuremath{\mathrm{BCH}}}

\newcommand{\Sel}{\ensuremath{\operatorname{Sel}}}
\newcommand{\deltaSel}{\ensuremath{\Sel_\delta}}

\newcommand{\ETR}{\textsc{ETR}}

\newcommand{\PosSLP}{\textrm{PosSLP}}
\newcommand{\BP}{\ensuremath{\operatorname{BP}}}

\newcommand{\Sol}{\ensuremath{\operatorname{Sol}}}

\newcommand{\MulC}{\ensuremath{\ast\zeta}}
\newcommand{\sqrtk}[1]{\ensuremath{\sqrt[\uproot{2}k]{#1}}}

\newcommand{\BasisPlusMinusTimes}{\ensuremath{\{+,-,\ast\}}}
\newcommand{\BasisPlusMinusTimesDiv}{\ensuremath{\{+,-,\ast,\div\}}}

\newcommand{\BasisPlusMulCMax}{\ensuremath{\{+,\MulC,\max\}}}
\newcommand{\BasisPlusMinusMulCMaxMin}{\ensuremath{\{+,-,\MulC,\max,\min\}}}

\newcommand{\BasisPlusTimesMax}{\ensuremath{\{+,\ast,\max\}}}
\newcommand{\BasisPlusMinusTimesMaxMin}{\ensuremath{\{+,-,\ast,\max,\min\}}}

\newcommand{\BasisTPlusTMinusTimesMaxMin}{\ensuremath{\{+_{T[0,1]},-_{T[0,1]},\ast,\max,\min\}}}

\newcommand{\BasisPlusTimesDivMax}{\ensuremath{\{+,\ast,\div,\max\}}}
\newcommand{\BasisPlusMinusTimesDivMaxMin}{\ensuremath{\{+,-,\ast,\div,\max,\min\}}}

\newcommand{\BasisPlusMinusTimesDivMaxMinRoots}{\ensuremath{\{+,-,\ast,\div,\max,\min,\sqrtk~\}}}

\newcommand{\NE}{\ensuremath{\mathrm{NE}}}

\newcommand{\boundary}{\ensuremath{\partial}}

\newtheorem{theorem}{Theorem}
\newtheorem{proposition}{Proposition}
\newtheorem{lemma}{Lemma}

\newtheorem{definition}{Definition}

\newtheorem{observation}{Observation}

\usepackage[
  backend=biber,
  style=alphabetic,
  sortlocale=en_US,
  useprefix=false,
  url=false,
  doi=true,
  eprint=false,
  maxalphanames=4,
  minalphanames=3,
  maxbibnames=99,
  uniquename=false, uniquelist=false
]{biblatex}
\begin{document}
\title{Strong Approximate Consensus Halving\\ and the Borsuk-Ulam Theorem}

\author[2]{Eleni Batziou}
\author[1]{Kristoffer Arnsfelt Hansen}
\author[1]{Kasper Høgh}
\affil[1]{Aarhus University\authorcr\email{arnsfelt@cs.au.dk}\\\email{kh@cs.au.dk}}
\affil[2]{Technical University of Munich\authorcr\email{batziou@in.tum.de}}

\maketitle

\begin{abstract}
  In the consensus halving problem we are given $n$ agents with
  valuations over the interval~$[0,1]$. The goal is to divide the
  interval into at most $n+1$ pieces (by placing at most $n$ cuts),
  which may be combined to give a partition of~$[0,1]$ into two sets
  valued equally by all agents. The existence of a solution may be
  established by the Borsuk-Ulam theorem.  We consider the task of
  computing an approximation of an exact solution of the consensus
  halving problem, where the valuations are given by distribution
  functions computed by algebraic circuits. Here approximation refers
  to computing a point that $\eps$-close to an exact solution, also
  called \emph{strong} approximation. We show that this task is
  polynomial time equivalent to computing an approximation to an exact
  solution of the Borsuk-Ulam search problem defined by a continuous
  function that is computed by an algebraic circuit.

  The Borsuk-Ulam search problem is the defining problem of the
  complexity class~$\BU$. We introduce a new complexity class~$\BBU$
  to also capture an alternative formulation of the Borsuk-Ulam
  theorem from a computational point of view. We investigate their
  relationship and prove several structural results for these classes
  as well as for the complexity class $\FIXP$.
\end{abstract}

\section{Introduction}
Many computational problems, e.g.\ linear and semidefinite
programming, are most naturally expressed using real numbers. When the
model of computation is discrete, these problems must be recast as
discrete problems. In the case of linear programming this causes no
problems.  Namely, when the input is given as rational numbers and an
optimal solution exists, a rational valued optimal solution exists and
may be computed in polynomial time. For semidefinite programming
however, it may be the case that all optimal solutions are
irrational. For dealing with such cases we may instead consider the
\emph{weak optimization} problem as defined by Grötschel, Lovász and
Schrijver~\cite{GrotschelLS1988}: Given $\eps>0$, the task is to
compute a rational-valued vector $x$ that is $\eps$-close to the set
of feasible solutions and has objective value $\eps$-close to
optimal. Assuming we are also given, as an additional input, a
strictly feasible solution and a bound on the magnitude of the
coordinates of an optimal solution, the weak optimization problem may
be solved in polynomial time using the ellipsoid
algorithm~\cite{GrotschelLS1988}. Let us note however that without
additional assumptions, even the complexity of the basic
\emph{existence problem} of semidefinite feasibility is unknown. In
fact, the problem is likely to be computationally very
hard~\cite{DAM:TarasovV2008}. More precisely, it is hard for the
problem $\PosSLP$, which is the fundamental problem of deciding
whether an integer given by a division free arithmetic circuit is
positive~\cite{SICOMP:AllenderBKM2009}.

In this paper we consider real valued search problems, where existence
of a solution is guaranteed by topological existence theorems such as
the Brouwer fixed point theorem and the Borsuk-Ulam theorem. This
means that the search problems are \emph{total}, thereby fundamentally
differentiating them from general search problems where, as described
above, even the existence problem may be computational hard. We are
mainly interested in the \emph{approximation} problem: given $\eps>0$,
the task is to compute a rational-valued vector $x$ that is
$\eps$-close to the set of solutions.

Recall that the Brouwer fixed point theorem states every continuous
function $f: B^n \rightarrow B^n$, where $B^n$ is the unit $n$-ball,
has a fixed point, i.e.\ there is $x \in B^n$ such that
$f(x)=x$~\cite{MA:Brouwer1911}. The Borsuk-Ulam theorem states that
every continuous function $f\colon S^n \rightarrow \RR^n$, where $S^n$
is the unit $n$-sphere in $\RR^{n+1}$, maps a pair of antipodal points
of $S^n$ to the same point in~$\RR^n$, i.e. there is $x\in S^n$ such
that $f(x)=f(-x)$~\cite{FM:Borsuk1933}. The Brouwer fixed point
theorem is of course not restricted to apply to the domain $B^n$, but
applies to any domain that is homeomorphic to~$B^n$. Similarly the
Borsuk-Ulam theorem applies to any domain homeomorphic to~$S^n$ by an
antipode-preserving homeomorphism. It is well-known that the
Borsuk-Ulam theorem generalizes the Brouwer fixed point theorem, in the
sense that the Brouwer fixed point theorem is easy to prove using the
Borsuk-Ulam theorem~\cite{AMM:Su1997,AMM:Volovikov08}.

The Brouwer fixed point theorem and the Borsuk-Ulam theorem naturally
define corresponding real valued search problems, and thereby also
corresponding approximation problems. In addition, the statements of
the theorems naturally leads to another notion of approximation. For
the case of the Brouwer fixed point theorem we may look for an
\emph{almost} fixed point, i.e.\ $x \in B^n$ such that $f(x)$ is
$\eps$-close to $x$, and for the case of the Borsuk-Ulam theorem we
look for a pair of antipodal points that \emph{almost} map to the same
point, i.e.\ $x \in S^n$ such that $f(x)$ and $f(-x)$ are
$\eps$-close. Following~\cite{SICOMP:EtessamiY10}, we shall refer to
this notion of approximation as \emph{weak} approximation and to make
the distinction clear we refer to the former (and general) notion of
approximation as \emph{strong} approximation. In the setting of weak
approximation in relation to the Borsuk-Ulam theorem we assume that
$f$ has co-domain $B^n$.

In their seminal work, Etessami and
Yannakakis~\cite{SICOMP:EtessamiY10} introduced the complexity class
$\FIXP$ to capture the computational complexity of the real-valued
search problems associated with the Brouwer fixed point theorem, and
proved that the problem of finding a Nash equilibrium in a given
3-player game in strategic form is \FIXP-complete. In order to have a
notion of completeness, the class $\FIXP$ is defined to be closed
under reductions. The type of reductions chosen by Etessami and
Yannakakis, SL-reductions, consists of mapping between sets of
solutions by a composition of a \emph{projection} reduction followed
by individual affine transformation applied to each coordinate.

Etessami and Yannakakis consider different ways to cast real valued
search problems as discrete search problems. In addition to the
approximation problem, these are the \emph{partial computation}
problem where the task is to compute a solution to a given number of
bits of precision and \emph{decision} problems, where the task is to
evaluate a sign condition of the set of solutions given the promise
that either all solutions satisfy the condition or none of them do.
Of these we shall only consider the approximation problem. The class
$\FIXPA$ denotes the class of discrete search problems corresponding to
strong approximation of Brouwer fixed points and is defined to be
closed under polynomial time reductions. Etessami and Yannakakis also
prove that the problem \PosSLP\ reduce to the problem of approximating
a Nash equilibrium, thereby showing that $\FIXPA$ likely contains
search problems that are computationally very hard.

While the notion of SL-reductions is very restricted, it is sufficient
for proving completeness of the problem of finding Nash
equilibrium. Likewise, SL-reductions are sufficient for showing that
$\FIXP$ is robust with respect to the choice of domain for the Brouwer
function.

Another important reason for using SL-reductions is that they
immediately imply polynomial time reductions between the corresponding
decision and approximation problems (the partial computation problem
is more fragile and requires additional assumptions,
cf.~\cite{SICOMP:EtessamiY10}). As we are mainly interested in the
approximation problem more expressive notions of reducibility can be
considered, while maintaining the property that reducibility implies
polynomial time reducibility between the corresponding approximation
problems. A sufficient condition for this is that the mapping of
solutions is \emph{polynomially continuous} and polynomial time
computable.

\subsection{The Borsuk-Ulam Theorem}
Deligkas, Fearnley, Melissourgos, and
Spirakis~\cite{JCSS:DeligkasFMS21} recently introduced a complexity
class $\BU$ to capture, in an analogy to $\FIXP$, the computational
complexity of the real-valued search problems associated with the
Borsuk-Ulam theorem.

The Borsuk-Ulam theorem has a number of equivalent statements that are
also easy to derive from each other.  A function $f$ defined on the
unit sphere $S^n$ is \emph{odd} if $f(x)=-f(-x)$ for all $x \in
S^n$. Note that the boundary $\boundary B^n$ of the unit $n$-ball
$B^n$ is identical to $S^{n-1}$. We thus say that a function $f$
defined on $B^n$ is odd on $\boundary B^n$ if $f$ is odd when
restricted to $S^{n-1}$. We present the simple proof of the known fact
that the different formulations can be derived from each other, for
the purpose of discussing equivalence from a computational point of
view.
\begin{theorem}[Borsuk-Ulam]
\label{THM:Borsuk-Ulam}
  The following statements hold:
  \begin{enumerate}[(1)]
  \item If $f\colon S^n \rightarrow \RR^n$ is continuous there exists
    $x \in S^n$ such that $f(x)=f(-x)$.
  \item If $g\colon S^n \rightarrow \RR^n$ is continuous and odd there
    exists $x \in S^n$ such that $g(x)=0$.
  \item If $h\colon B^n \rightarrow \RR^n$ is continuous and odd on
    $\boundary B^n$ there exists $x \in B^n$ such that \mbox{$h(x)=0$}.
  \end{enumerate}
\end{theorem}
\begin{proof}[Proof of equivalence]
  Given $f$ we may define $g(x)=f(x)-f(-x)$. Clearly $g$ is odd and we
  have $g(x)=0$ if and only if $f(x)=f(-x)$, which shows that (2)
  implies (1). Conversely, given $g$ we simply let $f=g$. If
  $f(x)=f(-x)$, then since $g$ is odd we have
  $f(x)=g(x)=-g(-x)=-f(-x)=-f(x)$ and hence $g(x)=f(x)=0$, which
  therefore shows (1) implies (2).

  We may view $S^n$ as two hemispheres, each homeomorphic to $B^n$,
  which are glued together along their equators. Let
  $\pi \colon S^n \rightarrow B^n$ be the orthogonal projection
  defined by $\pi(x_1,\dots,x_{n+1})=(x_1,\dots,x_n)$. Then given $h$
  we may define
  \[
    g(x) = \begin{cases} h(\pi(x)) & \text{ if } x_{n+1}\geq 0 \\
      -h(-\pi(x)) & \text{ if } x_{n+1}\leq 0 \end{cases} \enspace .
  \]
  The assumption that $h$ is odd on $\boundary B^n$ makes $g$ a
  well-defined continuous odd function. We have $g(x)=0$ if and only
  if $h(x)=0$, which shows that (2) implies (3). Conversely, given $g$
  we define $h$ by
  $h(x)=g\left(x,(1-\norm{x}_2^2)^{\tfrac{1}{2}}\right)$. Then $h$ is
  continuous and odd on $\boundary B^n$, since $x \in \boundary B^n$
  if and only if $\norm{x}^2_2=1$. Clearly if $h(x)=0$ we may let
  $y=(x,(1-\norm{x}_2^2)^{\tfrac{1}{2}})$ and have $g(y)=0$. On the
  other hand, when $g(y)=0$ we may define $x=(y_1,\dots,y_n)$ if
  $y_{n+1}\geq 0$ and $x=(-y_1,\dots,-y_n)$ if $y_{n+1} < 0$, and we
  have $h(x)=0$. Together this shows that (3) implies (2).
\end{proof}

The class $\BU$ defined in~\cite{JCSS:DeligkasFMS21} corresponds to
first formulation of the above theorem. We may clearly consider the
second formulation equivalent to the first also from a computational
point of view. In particular, when translating between formulations,
the set of solutions is unchanged. Note that this set of solutions has
the property that all solutions come in pairs: when $x$ is a solution
then $-x$ is a solution as well. For the third formulation of the
theorem this property only holds for solutions on the
boundary~$\boundary B^n$.

In contrast, while the mapping of solutions of the third formulation
to the second (and first) formulation given above is continuous this
is not the case in the other direction. More precisely, consider
$y \in S^n$ such that $g(y)=0$. For a solution strictly contained in
the upper hemisphere, the orthogonal projection to the first~$n$
coordinates produces $x \in B^n$ such that $h(x)=0$. For a
solution~$y$ strictly contained in the lower hemisphere, the
projection is instead applied to the antipodal solution~$-y$.

To clarify this issue from a computational point of view we introduce
a new class $\BBU$ of real valued search problems corresponding to the
third formulation of Theorem~\ref{THM:Borsuk-Ulam}, and it will follow
from definitions that $\BU \subseteq \BBU$. In the context of strong
approximation however, the corresponding classes of discrete search
problems $\BUA$ and $\BBUA$ will be shown to coincide. The idea is
that given an approximation to $y \in S^n$, where $g(y)=0$, that is
sufficiently close to the equator of $S^n$, there is no harm in
\emph{incorrectly} deciding to which hemisphere $y$ belongs, since
solutions $x \in \boundary B^n$ for which $h(x)=0$ also come in pairs.

For the class $\BU$, the notion of SL-reductions is clearly too
restrictive to allow a reasonable comparison to $\FIXP$. Closing the
class $\BU$ by SL-reductions, the solutions would still come in pairs,
thereby imposing strong conditions on the set of solutions. On the
other hand the reductions should also not be \emph{too} strong. In
particular it would be desirable that $\FIXP$ would be still be closed
under the chosen notion of reductions. This issue is not discussed
in~\cite{JCSS:DeligkasFMS21}. We shall therefore propose a suitable
notion of reductions for both $\BU$ and $\BBU$.
 
\subsection{Consensus Halving}
The Consensus halving problem is a classical problem of \emph{fair
  division}~\cite{MSS:SimmonsS03}. We are given a set of $n$ bounded
and continuous measures $\mu_1,\dots,\mu_n$ defined on the interval
$A=[0,1]$.  The goal is to partition the interval $A$ into at most
$n+1$ intervals, i.e.\ by placing at most $n$ \emph{cuts}, such that
unions of these intervals form another partition $A=A^+ \cup A^-$ of
$A$ satisfying $\mu_i(A^+)=\mu_i(A^-)$ for every $i$. We may think of
the intervals being assigned a \emph{label} from the set $\{+,-\}$,
and $A^+$ is precisely the union of the intervals labeled by~$+$. Such
a partition is also known as a consensus halving. Using the Borsuk-Ulam
theorem, Simmons and Su~\cite{MSS:SimmonsS03} proved that a consensus
halving using at most $n$ cuts always exists. Simmons and Su represent
a division of $A$ as a point $x$ on the unit $n$-sphere $S^n_1$ with
respect to the $\ell_1$-norm. The point $x$ is viewed as representing
a division into \emph{precisely} $n+1$ intervals, where some intervals
are possibly empty. More precisely, the $i$-th interval has length
$\abs{x_i}$, and intervals of length~0 may simply be discarded. The
intervals of positive length are then labeled according
to~$\sgn(x_i)$. Note that for any $x$, the antipode $-x$ represent the
division where the sets $A^+$ and $A^-$ are exchanged. This naturally
leads to a formulation using the Borsuk-Ulam
theorem~\cite{MSS:SimmonsS03}. Namely we may consider the function
$F \colon S^n_1 \rightarrow \RR^n$ given by $F(x)_i= \mu_i(A^+)$, and
note that any $x \in S^n_1$ for which $F(x)=F(-x)$ represent a
consensus halving.

We are interesting in the simple setting of additive measures, where we
have corresponding density functions $f_1,\dots,f_n$ such that
$\mu_i(B) = \int_B f_i(x)\, dx$. To cast the consensus halving problem
as a real valued search problem we follow~\cite{JCSS:DeligkasFMS21}
and assume that the measures $\mu_1,\dots,\mu_n$ are given by the
distribution functions $F_1,\dots,F_n$ defined by
$\int_0^x f_i(x)\, dx$. An instance of the consensus halving problem
is then given as a list of algebraic circuits computing these
distribution functions.

Corresponding to the different formulations of the Borsuk-Ulam theorem
as a real valued search problem with domain $S^n$ or $B^n$ we get two
different formulations of the consensus halving problem. We denote
these by $\CH$ and $\BCH$ respectively. Deligkas~et~al.\ proved
membership of $\CH$ in $\BU$ following the proof of Simmons and Su,
and proved hardness of $\CH$ for $\FIXP$. Combining these, it follows
that $\FIXP \subseteq \BU$.

\subsection{Strong versus Weak Approximation}
The difference between weak and strong approximation was studied in
detail in the general context of the Brouwer fixed point theorem by
Etessami and Yannakakis. A central example is the problem of finding a
Nash equilibrium (NE). An important notion of approximation of a NE is
the notion of an $\eps$-NE. Computing an $\eps$-NE of a given
strategic form game $\Gamma$ is polynomial time equivalent to
computing a weak $\eps'$-approximation to a fixed point the Nash's
Brouwer function $F_\Gamma$ associated to
$\Gamma$~\cite[Proposition~2.3]{SICOMP:EtessamiY10}. In turn,
computing a weak $\eps'$-approximation to a fixed point of $F_\Gamma$
polynomial time reduces to computing a strong $\eps''$-approximation
to a fixed point of
$F_\Gamma$~\cite[Proposition~2.2]{SICOMP:EtessamiY10}, since the
function $F_\Gamma$ is polynomially continuous and polynomial time
computable. In general however an $\eps$-NE might be far from any
actual NE, unless $\eps$ is inverse doubly exponentially small as a function
of the size of the game~\cite[Corollary~3.8]{SICOMP:EtessamiY10}.

For the problem of consensus halving we can illustrate the difference
between weak and strong approximation by a simple example. We shall
refer to a weak $\eps$-approximation of a consensus halving as simply
an $\eps$-consensus halving. Consider a single agent whose measure
$\mu$ is on the interval $[0,1]$ is given by the following density
\[
  f(x) = \begin{cases}
    (1+\eps)/\eps & \text{ if } 0\leq x < \eps/2\\
    0 & \text{ if } \eps/2 \leq x < 1-\eps/2\\
    (1-\eps)/\eps & \text{ if } 1-\eps/2 \leq x \leq 1
  \end{cases}
\]
We have $\mu([0,1])=1$ and since $\mu$ is a step function, the
corresponding distribution function $F$ is piecewise linear. The
unique consensus halving is obtained by placing a cut at the point
$\eps/2-\eps^2/(2+2\eps)$. Placing a cut at any point
$t \in [\eps/2-\eps^2/(1+\eps),1-\eps/2]$ results in an
$\eps$-consensus halving, i.e.\ such that
$\abs{\mu([0,t])-\mu([t,1])}\leq \eps$. Thus an $\eps$-consensus
halving might be very far from an actual consensus halving. Note also
that placing a cut at any point $t \in [0,3\eps/2-\eps^2/(2+2\eps)]$
is a strong $\eps$-approximation, which illustrates that a strong
approximation is not necessarily a weak approximation. On the other
hand, a strong $(\eps^2/2)$-approximation is also an $\eps$-consensus
halving.

The Brouwer fixed point theorem and the Borsuk-Ulam theorem can both
be proved starting from combinatorial analogoues of the two
theorems, namely from Sperner's lemma and Tucker's lemma,
respectively. The proofs of these two lemmas are constructive, but
using them to derive the Brouwer fixed point theorem and the
Borsuk-Ulam theorem involve a nonconstructive limit argument. Let us
in passing note that while Sperner's lemma, like the Borsuk-Ulam
theorem, has several different formulations, it is usually formulated
as the combinatorial analogue of the third formulation of
Theorem~\ref{THM:Borsuk-Ulam}.

Sperner's and Tucker's lemma give rise to total $\NP$ search
problems. These turn out to be complete for the complexity classes
$\PPAD$ and $\PPA$ introduced in seminal work by
Papadimitriou~\cite{JCSS:Papadimitriou1994}. Papadimitriou proved
$\PPAD$-completeness of the problem given by Sperner's lemma as well
as membership in $\PPA$ of the problem given by Tucker's lemma, while
$\PPA$-completeness of the latter problem was proved recently by
Aisenberg, Bonet, and Buss~\cite{JCSS:AisenbergBB20}. These results
also imply that the classes $\PPAD$ and $\PPA$ corresponds to the
problems of computing weak approximations to Brouwer fixed points and
to Borsuk-Ulam points.

The computational complexity of the problems of computing an $\eps$-NE
and of computing an $\eps$-consensus halving was settled in
breakthroughs of two lines of research. Computing an $\eps$-NE was
shown to be $\PPAD$-complete by Daskalakis, Goldberg and
Papadimitriou~\cite{SICOMP:DaskalakisGP2009} and Cheng and
Deng~\cite{FOCS:ChenDeng2006}. Computing an $\eps$-consensus halving
was shown to be $\PPA$-complete by Filos-Ratsikas and
Goldberg~\cite{STOC:Filos-RatsikasG18,STOC:Filos-RatsikasG19}.

\subsection{Our Results}
Our main result is that the problem of strong approximation of
consensus halving is equivalent to strong approximation of the
Borsuk-Ulam theorem.
\begin{theorem}
  \label{THM:CH-BUa-complete}
  The strong approximation problem for $\CH$ is $\BUA$-complete.
\end{theorem}
As described we view the consensus halving problem as the real valued
search problem with its domain being either the unit sphere or the
unit ball with respect to the $\ell_1$-norm. The theorem is proved by
reduction from the real valued search problem associated with the
Borsuk-Ulam theorem on the domain being the unit ball with respect to
the $\ell_\infty$-norm, i.e.\ from a defining problem of the class
$\BBU$.

It is of general interest to study the relationship between search
problems given by the Borsuk-Ulam theorem on different
domains from a computational point of view. The reduction establishing
the proof of Theorem~\ref{THM:CH-BUa-complete} gives additional
motivation for this. The domains we consider are unit spheres $S^n_p$
and unit balls $B^n_p$ with respect to the $\ell_p$-norm for $p\geq 1$
or $p=\infty$. It is of course straightforward to construct
homeomorphisms between unit spheres or unit balls with respect to
different norms, and these could be used to define reductions between
the different problems. We would however like that the mapping of
solutions is simple, and in particular we would like to avoid
divisions and root operations. We prove that one may in fact reduce
between domains using SL-reductions.

Deligkas~et~al.\ gave a reduction from the \FIXP-complete problem of
finding a Nash equilibrium to \CH. Combined with membership of $\CH$
in $\BU$, this gives the inclusion $\FIXP \subseteq \BU$. We observe
that a proof due to Volovikov~\cite{AMM:Volovikov08} of the Brouwer
fixed point theorem from the Borsuk-Ulam theorem may be adapted to
give a simple proof of the inclusion $\FIXP \subseteq \BU$.

For the class $\FIXP$ we prove two interesting structural properties
that do not appear to have been observed earlier. While $\FIXP$ is
defined using SL-reductions, we show that $\FIXP$ is closed under
polynomial time reductions where the mapping of solutions is expressed
by \emph{general} algebraic circuits. This in particular supports that
one may reasonably define the classes $\BU$ and $\BBU$ using less
restrictive notions of reductions than SL-reductions. We propose to
have the mapping of solutions be computed by algebraic circuits
involving the operations of addition, multiplication by scalars, as
well as maximization. This means that the mapping of solutions is a
piecewise linear function, and we refer to these as PL-reductions. The
second structural result for $\FIXP$ is a characterization of the
class by very simple Brouwer functions. These are defined on the
unit-hypercube domain $[0,1]^n$ and each coordinate function is simply
one of the operations $\BasisPlusMinusTimesMaxMin$, modified to be
have the output truncated to the interval $[0,1]$.

For the classes $\BU$ and $\BBU$ we prove that they are also closed
under reductions where the mapping of solutions is computed by general
algebraic circuits, but with the additional requirement that this
function must be odd.

For the class $\FIXP$, an interesting consequence of the proof that
finding a Nash equilibrium is complete, is that the class may be
characterized by Brouwer functions computed by algebraic circuits
without the division operation. The proof also shows that the class
$\FIXP$ is unchanged even when allowing root operations as basic
operations. We prove by a simple transformation that the classes $\BU$
and $\BBU$ may be characterized using algebraic circuits without the
division operation. Furthermore, as a consequence of
Theorem~\ref{THM:CH-BUa-complete} the class of strong approximation
problems $\BUA=\BBUA$ is unchanged even when allowing root operations
as basic operations.

\subsection{Comparison to previous work}
As a precursor to the proof of $\PPA$-completeness of computing an
$\eps$-consensus halving, Filos{-}Ratsikas, Frederiksen, Goldberg and
Zhang~\cite{MFCS:Filos-RatsikasFGZ18} proved the problem to be
$\PPAD$-hard. Deligkas~et~al.~\cite{JCSS:DeligkasFMS21} uses ideas
from this proof together with additional new ideas to obtain their
proof of $\FIXP$-hardness for computing an exact consensus halving.

While $\PPAD \subseteq \PPA$, the $\PPAD$-hardness result of
\cite{MFCS:Filos-RatsikasFGZ18} is not implied by the recent proofs of
$\PPA$-completeness. In particular, the work
\cite{MFCS:Filos-RatsikasFGZ18} proves $\PPAD$-hardness even for
\emph{constant} $\eps$, while the work of
\cite{STOC:Filos-RatsikasG19} only proves $\PPA$-hardness for $\eps$
being inverse polynomially small. In the same way, while
$\FIXP \subseteq \BU$, $\FIXP$-hardness of computing an exact
consensus halving is not implied by our reduction, since
Theorem~\ref{THM:CH-BUa-complete} establishes $\BUA$-hardness rather
than $\BU$-hardness. Recently a considerably simpler proof of
$\PPA$-hardness for computing an $\eps$-consensus halving was given by
Filos{-}Ratsikas, Hollender, Sotiraki and
Zampetakis~\cite{EC:Filos-RatsikasHSZ20}, and our reduction is
inspired by this work.

All reductions described above are similar in the sense that one or
more evaluations of a circuit are expressed in the consensus halving
instance. The full interval $A$ is partitioned into subintervals, 
cuts within these subintervals encode values in various ways, and
agents implement the gates of the circuit by placing cuts. A main
difference between the reductions establishing $\PPAD$-hardness and
$\FIXP$-hardness to those establishing $\PPA$-hardness is that in the
former reductions, all cuts are constrained to be placed in distinct
subintervals. This reason this is possible is that the objective is to
find a fixed point of the circuit, which means that inputs and outputs
may be identified.

In the setting of $\PPA$ and $\BBU$ the objective is to find a
``zero'' of the circuit. More precisely, for the setting of $\PPA$ the
objective is to find two adjacent points of a given Tucker labeling
that receive complementary labels, i.e.\ labels of different sign but
same absolute value. For the setting of $\BBU$ the objective is to
find an actual zero point of the circuit. All of the reductions
establishing $\PPA$-hardness of computing an $\eps$-consensus halving
have the property that cuts encoding the input of the circuit are
\emph{free} cuts, meaning that they can in principle be placed
anywhere, and as a result also interfere with the evaluations of the
circuit. This is also the case for our reduction, and this invariably
limits its applicability to the approximation problem.

In the reduction of~\cite{EC:Filos-RatsikasHSZ20}, the interval $A$ is
structured into different regions, a coordinate-encoding region, a
constant-creation region, several circuit-simulation regions, and
finally a feedback region. Our reduction also has a
coordinate-encoding region and several circuit simulation regions, but
the functions performed by the constant-creation region and feedback
regions perform in~\cite{EC:Filos-RatsikasHSZ20} is our reduction
integrated in the individual circuit simulation regions and done
differently.

A novelty of the reduction of~\cite{EC:Filos-RatsikasHSZ20} compared
to previous reductions is in how values are encoded by cuts in
subintervals. In previous reductions, values are encoded by what we
will call \emph{position encoding}. Here it is required that there is
exactly one cut in the subinterval, and the value encoded is
determined by the distance between the cut position and the left
endpoint of the interval. In~\cite{EC:Filos-RatsikasHSZ20} values are
encoded by what we will call \emph{label encoding}. Here there is no
requirement on the number of cuts in the subinterval, and the value
encoded is simply the difference between the Lebesgue measures of the
subsets of the interval receiving label~$+$ and label~$-$. We shall
employ a hybrid approach where the coordinate-encoding region uses
label encoding while the circuit-simulation regions uses position
encoding. The first step performed in a circuit-simulation region is
thus to copy the input from the coordinate-encoding region. Switching
to position encoding allows us in particular to implement a
multiplication gate, similarly to~\cite{JCSS:DeligkasFMS21}. Here
the multiplication $xy$ is computed via the identity
$xy=((x+y)^2-x^2-y^2)/2$. In~\cite{JCSS:DeligkasFMS21} where values
range over~$[0,1]$, the squaring operation may be implemented directly
by agents. In our case values range over the interval~$[-1,1]$, and
the squaring operation is decomposed further, having agents compute it
separately over the intervals~$[-1,0]$ and~$[0,1]$.

In analogy to~\cite{EC:Filos-RatsikasHSZ20} we have feedback agents
that ensures that the circuit evaluates to~0 on the encoded input. The
criteria that the agents check is however different, and for our
purposes it is crucial that we have the same sign pattern in the
position encoding of the output of the circuit as the copy of the
input made by the circuit-simulation region. The actual detection of
an output of~0 is performed by using approximations of the Dirac delta
function. For computing the distribution functions of the feedback
agents, we make use of the fact that these are computed by algebraic
circuits, which enable us to make a strong approximation of the Dirac
delta function via repeated squaring.

\subsection{Organization of Paper}
In Section~\ref{SEC:Preliminaries} we introduce necessary terminology
and we give a detailed account of real valued search problems and
reducibility between these. Our structural results for $\FIXP$ are
given in Section~\ref{SEC:Structural-FIXP} and our structural results
for $\BU$ and $\BBU$ are given in
Section~\ref{SEC:BU-BBU}. Section~\ref{SEC:BU-BBU} also includes the
simple proof of the inclusion $\FIXP \subseteq \BU$. We present our
main result, Theorem~\ref{THM:CH-BUa-complete}, in
Section~\ref{SEC:BBU-to-CH}.

\section{Preliminaries}
\label{SEC:Preliminaries}

\subsection{Algebraic Circuits}
Let $B$ be a finite set of real valued functions, for example
$B=\BasisPlusMinusTimesDivMaxMin$. An \emph{algebraic circuit} $C$
with $n$ inputs and $m$ outputs over the \emph{basis} $B$ is given by
an acyclic graph $G=(V,A)$ as follows. The \emph{size} of $C$ is equal
to the number of nodes of $G$, which are also referred to as
\emph{gates}. The \emph{depth} of $C$ is equal to the length of the
longest path of $G$. Every node of indegree~0 is either an \emph{input
  gate} labeled by a variable from the set $\{x_1,\dots,x_n\}$ or a
\emph{constant gate} labeled by a real valued constant. Every other
node is labeled by an element of $B$ called the \emph{gate
  function}. If a node $v$ is labeled by a gate function
$g \colon A \rightarrow \RR$ with $A \subset \RR^k$ we require that $g$ has
exactly $k$ ingoing arcs with a linear order specifying the order of
arguments to $g$. The output of $C$ is specified by an ordered list of
$m$ (not necessarily distinct) nodes of $G$. The computation of $C$ on
a given input $x \in \RR^n$ is defined in the natural way. Computation
may fail in case a gate of $C$ labeled by a function
$g\colon A \rightarrow \RR$ receives an input outside $A$, and in this case
the output of $C$ is undefined. Otherwise we say that the output is
well defined and denote its value by $C(x)$.  If $D \subseteq \RR^n$
we say that $C$ computes a function $f\colon D \rightarrow \RR^m$ if $C(x)$
is well defined for all $x \in D$.

We shall in this paper just consider algebraic circuits where the
basis consists only of continuous functions. This means in particular
that any algebraic circuits computes a continous function as well. We
shall also only consider consider constant gates labeled with rational
numbers. In this case we are also interested in the \emph{bitsize} of
the encoding of the constants, which is the maximum bitsize of the
numerator or denominator. An important special class of algebraic
circuits are those over the basis $\BasisPlusMinusTimesDiv$ and using
just the constant~1. We refer to these as \emph{arithmetic
  circuits}. An arithmetic circuit with no division gates is called
division-free. Note that any integer of bitsize $\tau$ may be computed
by a division-free arithmetic circuit of size $O(\tau)$.

By using multiplication with the constant $-1$, the functions $-$ and
$\min$ may be simulated using $+$ and $\max$, respectively. In this
way we may convert a circuit over the full basis
$\BasisPlusMinusTimesDivMaxMin$ into an equivalent
$\BasisPlusTimesDivMax$-circuit. We shall also consider circuits where use
of the multiplication operator $\ast$ is \emph{restricted} to having
one of the arguments being a constant gate. We denote this by the
symbol $\MulC$ and use it in particular for defining
$\BasisPlusMulCMax$-circuits.

At times it will convinient to consider gate functions with their
output range truncated to stay within a given interval. If
$g\colon A \rightarrow \RR$ is a gate function and $a\leq b$ defines a real
interval $[a,b]$ we denote by $g_{T[a,b]}$ the gate function defined
by $g_{T[a,b]}(x)=a$ if $g(x)<a$, $g_{T[a,b]}(x)=b$ if $g(x)>b$, and
$g_{T[a,b]}(x)=g(x)$ otherwise. Note that $g_{T[a,b]}$ is continuous
whenever $g$ is continuous.

While we shall not consider circuits with the discontinous sign
function $\sgn$, in the context of approximating functions, it is
sometimes sufficient to use an approximation of $\sgn$ instead.  A
typical use of $\sgn(z)$ is to perform a selection between two values
$x$ and $y$.  We define the $\delta$-approximate selection function to
be the function that based on $\sgn(z)$ outputs values $x$ or $y$
except in the interval of length $\delta$ centered around~$0$
where it instead linearly interpolates between $x$ and $y$.
\begin{definition}
  For given $\delta>0$, the (two-sided) $\delta$-approximate selection
  function $\Sel$ is defined by
  \[
  \deltaSel(x,y,z) =
  \begin{cases}
    x & \text{ if } z \leq -\delta/2 \\  (y-x)z/\delta+(y+x)/2 & \text{
      if } -\delta/2 \leq z \leq \delta/2\\ y & \text{ if } \delta/2 \leq z
  \end{cases}
  \]
\end{definition}
We note that $\deltaSel$ may be computed as
$\deltaSel(x,y,z)=(1-t)/2\cdot x+(1+t)/2\cdot y$, where $t$ defined by
$t=\max(\min(z,\delta/2),-\delta/2)/(\delta/2)$ is the
$\delta$-approximation of $\sgn(z)$. In particular is
$\deltaSel(x,y,z)$ computed by a $\BasisPlusTimesMax$-circuit (or a
$\BasisPlusTimesDivMax$-circuit if we also view $\delta$ as a
variable).

\subsection{Search problems}
A general search problem $\Pi$ is defined by specifying to each input
instance $I$ a \emph{search space} (or \emph{domain}) $D_I$ and a set
$\Sol(I) \subseteq D_I$ of \emph{solutions}. We distinguish between
\emph{discrete} and \emph{real-valued} search problems. For discrete
search problems we assume that $D_I \subseteq \zo^{d_I}$ for an
integer $d_I$ depending on $I$. Analogously, for real-valued search
problems we assume that $D_I \subseteq \RR^{d_I}$ for an integer
depending on $I$. One could likewise distinguish between search problems
with discrete input and real-valued input. We are however mostly
interested in problems where the input is discrete, that is we assume
that instances $I$ are encoded as strings over a given finite alphabet
$\Sigma$ (e.g.\ $\Sigma=\zo$).

A very important class of discrete search problems arise from decision
problems given as languages in $\NP$, thereby forming the class of
$\NP$ search problems. More precisely, these are the discrete search
problems where we assume there are polynomial time algorithms that
(i)~given $I$ compute $d_I$ whose magnitude is polynomial in $\abs{I}$,
(ii)~given $I$ and $x \in \zo^{d_I}$ checks whether $x \in D_I$, and
lastly, (iii)~given $I$ and $x \in D_I$ checks whether $x \in \Sol(I)$. The
corresponding language in $\NP$ is then
$L=\{I \mid \Sol(I) \neq \emptyset\}$. The class of all $\NP$ search
problems is denoted by $\FNP$. The subclass $\TFNP$ of $\FNP$ consists
of the $\NP$ search problems for which $\Sol(I)\neq \emptyset$ for
every input $I$. An $\NP$ search problem $\Pi$ is said to be solvable
in polynomial time if there is a Turing machine running in polynomial
time that on input $I$ gives as output \emph{some} member $y$ of
$\Sol(I)$ in case $\Sol(I)\neq \emptyset$ and rejects otherwise. The
subclass of $\FNP$ consisting of the search problems solvable in
polynomial time is denoted by $\FP$, and it holds that $\FP=\FNP$ if
and only if $\PTIME=\NP$.

Many natural search problems are however defined with a continous
search space. Not all of these may adequately be recast as discrete
search problems, but are more naturally viewed as real-valued search
problems. One approach for studying such problems would be to switch
to the Blum-Shub-Smale model of computation~\cite{BAMS:BlumSS1989}. A
BSS machine resembles a Turing machine, but operates with real numbers
instead of symbols from a finite alphabet. In particular is the input
real-valued, and input instances are therefore encoded as real-valued
vectors. All basic arithmetic operations and comparisons are unit-cost
operations. One may then define real-valued analogues of Turing
machine based classes. In particular, Blum, Shub and Smale defined and
studied the real-valued analogues $\PTIME_\RR$ and $\NP_\RR$ of
$\PTIME$ and $\NP$. A BSS machine may in general make use of
real-valued \emph{machine constants}. If a BSS machine only uses
rational valued machine constants we shall call it
\emph{constant-free}. Real-valued analogoues of the classes $\FP$,
$\FNP$, and $\TFNP$ for the BSS machine model do not appear to be
defined in the literature, but can be defined in a straight-forward
way. Let us just note that the proof that $\PTIME=\NP$ implies
$\FP = \FNP$ does not generalize to the setting of BSS machines, since
it crucially depends on the search space being discrete.

For the classes $\PTIME_\RR$ and $\NP_\RR$, if we simply restrict the
input to be discrete and consider only constant-free BSS machines this
results in complexity classes, denoted by $\BP(\PTIME_\RR^0)$ and
$\BP(\NP_\RR^0)$, that may directly be compared to Turing machine
based complexity classes. Indeed, it was proved by Allender,
Bürgisser, Kjeldgaard{-}Pedersen and
Miltersen~\cite[Proposition~1.1]{SICOMP:AllenderBKM2009} that
$\BP(\PTIME^0_\RR)=\PTIME^\PosSLP$, where $\PosSLP$ is the problem of
deciding whether an integer given by a division free arithmetic
circuit is positive. While the precise complexity of $\PosSLP$ is not
known, Allender~et~al.\ proved that it is contained in the
\emph{counting hierarchy} $\CH$ (not to be confused with the consensus
halving problem whose abbreviation coincides).

The class $\BP(\NP_\RR^0)$ is equal to the class $\cETR$ that was
defined by Schaefer and Štefankovič~\cite{TOCS:SchaeferS2017}
to capture the complexity of the existential theory of the reals
$\ETR$. It is known that $\NP \subseteq \cETR \subseteq \PSPACE$,
where the latter inclusion follows from the decision procedure for
$\ETR$ due to Canny~\cite{STOC:Canny1988}. Schaefer and Štefankovič
also prove $\cETR$-completeness for deciding existence of a
probability-constrained Nash equilibrium in a given 3-player game in
strategic form; later works have extended this to $\cETR$-completeness
for many other decision problems about existence of Nash equilibria
satisfying different properties in 3-player games in strategic
form~\cite{TEAC:GargMVY2018,STACS:BiloM2016,STACS:BiloM17,SAGT:BerthelsenH19}. The
proofs of $\cETR$-hardness makes critical use of the fact that the
input is discrete and it is not known if these problems are also
complete for $\NP_\RR$.

We define the class of $\cETR$ search problems as the following
subclass of all real valued search problems. Instaces $I$ are encoded
as string over a given finite alphabet $\Sigma$ and we assume there is
a polynomial time algorithm that given $I$ computes $d_I$, where
$D_I \subseteq \RR^{d_I}$. We next assume that there are polynomial
time constant free \emph{BSS machines} that given $I$ and $x \in \RR^{d_I}$
checks whether $x \in D_I$, and given $I$ and $x \in D_I$ checks
whether $x \in \Sol(I)$. The corresponding language in $\cETR$ is then
$L = \{I \mid \Sol(I)\neq \emptyset\}$.

\subsection{Solving real-valued search problems}
Let $\Pi$ be a $\cETR$ search problem. In analogy with the case of
$\NP$ search problems, one could consider the task of solving $\Pi$ to
be that of giving as output some member $y$ of $\Sol(I)$ in case
$\Sol(I)\neq \emptyset$. In general each member of $\Sol(I)$ may be
irrational valued which precludes a Turing machine to compute a
solution explicitly. This is in general also the case for a BSS
machine, even when allowing machine constants. Regardless, we shall
restrict our attention to Turing machines below.

On the other hand, when $\Sol(I) \neq \emptyset$ a solution is
guaranteed to exist with coordinates being algebraic numbers, since a
member of $\Sol(I)$ may be defined by an existential first-order
formula over the reals with only rational-valued coefficients. This
means that one could instead compute an indirect description of the
coordinates of a solution, for instance by describing isolated roots
of univariate polynomials. If such a description could be computed in
polynomial time in $\abs{I}$ we could consider that to be a polynomial
time solution of $\Pi$.

Etessami and Yannakakis~\cite{SICOMP:EtessamiY10} suggest several
other computational problems one may alternatively consider in place
of solving a search problems $\Pi$ explicitly or exactly. Our main
interest is in the problem of \emph{approximation}. We shall assume
for simplicity that $D_I \subseteq [-1,1]^{d_I}$. Together with an
instance $I$ of $\Pi$ we are now given as an auxiliary input a
rational number $\eps>0$, and the task is to compute $x \in \QQ^{d_I}$
such that there exist $x^* \in \Sol(I)$ with
$\norm{x^*-x}_\infty \leq \eps$. We shall turn this into a
\emph{discrete} search problem by encoding the coordinates of $x$ as
binary strings. More precisely, to $\Pi$ we shall associate a discrete
search problem $\Pi_a$ where instances are of the form $(I,k)$, where
$I$ is an instance of $\Pi$ and $k$ is a positive integer. We define
$\eps=2^{-k}$ and let the domain of $(I,k)$ be
$D_{I,k} = \zo^{d_{I}(k+3)}$, thereby allowing the specification of a
point $x \in D_I$ with coordinates of the form $x_i = a_i2^{-{k+1}}$,
where $a_i \in \{-2^{k+1},\dots,2^{k+1}\}$. The solution set
$\Sol(I,k)$ is defined from $\Sol(I)$ by approximating each
coordinate. That is, we define
$\Sol(I,k) = \{x \in D_{I,k} \mid \exists x^* \in \Sol(I) :
\norm{x^*-x}_\infty \leq \eps\}$. Note that if we had defined
$\Sol(I,k)$ by instead truncating the coordinates of solutions
$x^* \in \Sol(I)$ to $k$ bits of precision, we would have obtained the
possibly harder problem of \emph{partial computation} which was also
considered by Etessami and Yannakakis~\cite{SICOMP:EtessamiY10}.

We say that $\Pi$ can be approximated in polynomial time if the
approximation problem $\Pi_a$ can be solved in time polynomial
in~$\abs{I}$ and~$k$.

\subsection{Reductions between search problems}

Let $\Pi$ and $\Gamma$ be search problems. A \emph{many-one reduction}
from $\Pi$ to $\Gamma$ consists of a pair of functions $(f,g)$. The
function $f$ is called the instance mapping and the function $g$ the
solution mapping. The instance mapping $f$ maps any instance $I$ of
$\Pi$ to an instance $f(I)$ of $\Gamma$ and for any solution
$y \in \Sol(f(I))$ of $\Gamma$ the solution mapping $g$ maps the pair
$(I,y)$ to a solution $x=g(I,y) \in \Sol(I)$ of $\Pi$. It is required
that $\Sol(f(I)) \neq \emptyset$ whenever $\Sol(I) \neq \emptyset$. We
will only consider many-one reductions, and will refer to these simply
as \emph{reductions}.

If $\Pi_1$ and $\Pi_2$ are discrete search problems a reduction
$(f,g)$ between $\Pi_1$ and $\Pi_2$ is a \emph{polynomial time
  reduction} if both functions $f$ and $g$ are computable in
polynomial time. If $\Pi_1$ and $\Pi_2$ are real-valued search
problems it is less obvious which notion of reduction is most
appropriate and we shall consider several different types of
reductions. For all these we assume that $f$ is computable in
polynomial time. The reduction $(f,g)$ is a \emph{real polynomial time
  reduction} if $g$ is computable in polynomial time by a constant
free BSS machine. We shall generally consider this notion of reduction
too powerful. In particular the definitioon does not guaranteed that
the function $g$ is a continuous function in its second argument
$y$. For this reason we instead consider reductions defined by
algebraic circuits over a given basis $B$ of real-valued basis
functions.

We say that the reduction $(f,g)$ is a \emph{polynomial time
  $B$-circuit reduction} if there is a function computable in polynomial
time thats maps an instance $I$ to a $B$-circuit $C_{I}$ in such a way
that $C_I$ computes a function $C_I\colon D_{f(I)} \rightarrow D_I$ where
$g(I,y)=C_I(y)$ for all $y \in \Sol(f(I))$. Note in particular that
the size of $C_I$ and the bitsize of all constant gates are bounded by
a polynomial in $\abs{I}$. If in addition there exists a constant $h$
such that the depth of $C_I$ is bounded by $h$ for all $I$ we say that
the reduction $(f,g)$ is a \emph{polynomial time constant depth
  $B$-circuit reduction}. Etessami and
Yannakakis~\cite{SICOMP:EtessamiY10} defined the even weaker notion
where the function $f$ is a \emph{separable linear} transformation.
The reduction $(f,g)$ is an SL-reduction if there is a function
$\pi \colon \{1,\dots,d_I\} \rightarrow \{1,\dots,d_{f(I)}\}$ and rational
constants $a_i,b_i$, for $i=1,\dots,d_I$, all computable in polynomial
time from $I$, such that for all $y \in \Sol(f(I))$ it holds that
$x_i=a_i y_{\pi(i)}+b_i$, where $x=g(I,y)$. Thus an SL-reduction is
simply a \emph{projection reduction} together with an individual
affine transformation of each coordinate of the solution.

Functions computed by algebraic circuits over the basis
$\BasisPlusMulCMax$ are piecewise linear. We shall thus call
polynomial time $\BasisPlusMulCMax$-circuit reductions for polynomial
time piecewise linear reductions, or simply PL-reductions.

It is easy to see that all notions of reductions defined above are
transitive, i.e.\ if $\Pi$ reduces to $\Gamma$ and $\Gamma$ reduces to
$\Lambda$, then $\Pi$ reduces to $\Lambda$ as well.

A desirable property of PL-reductions is that the solution mapping $g$
is \emph{polynomially continuous}. By this we mean that for all
rational $\eps>0$ there is a rational $\delta>0$ such that for all
points $x$ and $y$ of the domain, $\norm{x-y}_\infty \leq \delta$
implies $\norm{g(x)-g(y)}_\infty \leq \eps$, and the bitsize of
$\delta$ is bounded by a polynomial in the bitsize of $\eps$ and of
$\abs{I}$. An example of a notion of reductions not guaranteed to be
polynomially continuous would be $\BasisPlusTimesMax$-circuit
reductions, since a circuit might perform repeated squaring. However,
constant depth $\BasisPlusTimesMax$-circuit reductions would still be
polynomially continuous.

\subsection{Total real-valued search problems}
Like in the case of $\TFNP$ where interesting classes of total $\NP$
search problems may be defined in terms of existence theorems for
finite structures~\cite{JCSS:Papadimitriou1994,JCSS:GoldbergP2018}, we
may define classes of total real valued $\cETR$ search problems based
on existence theorems concerning domains $D_I \subseteq
\RR^n$. Typical examples of such domains $D_I$ are spheres and
balls. Suppose $p$ is either a real number $p\geq 1$ or $p=\infty$. By
$S^n_p$ and $B^n_p$ we denote the unit $n$-sphere and unit $n$-ball
with respect to the $\ell_p$-norm defined as
$S^n_p = \{x \in \RR^{n+1} \mid \norm{x}_p=1\}$ and
$B^n_p = \{x \in \RR^n \mid \norm{x}_p \leq 1\}$, respectively. If $p$
is not specified, we simply assume $p=2$.

\subsubsection{The Brouwer fixed point theorem and \FIXP}
We recall here the definition of the class $\FIXP$ by Etessami and
Yannakis~\cite{SICOMP:EtessamiY10}. The class $\FIXP$ is defined by
starting with $\cETR$ search problems given by the Brouwer fixed point
theorem, and afterwards closing the class with respect to
SL-reductions. We shall refer to these defining problems as
\emph{basic} $\FIXP$ problems.
\begin{definition}
  An $\cETR$ search problem $\Pi$ is a \emph{basic} $\FIXP$ problem if
  every instance $I$ describes a nonempty compact convex domain $D_I$
  and a continuous function $F_I \colon D_I \rightarrow D_I$, computed by
  an algebraic circuit $C_I$, and these descriptions must be
  computable in polynomial time. The solution set is
  $\Sol(I) = \{x \in D_I \mid F_I(x)=x\}$.
\end{definition}
The Brouwer fixed point theorem guarantees that every basic $\FIXP$
problem is a total $\cETR$ search problem. To define the class
$\FIXP$, Etessami and Yannakis restrict attention to a concrete class
of basic $\FIXP$ problems.
\begin{definition}
  The class $\FIXP$ consists of all \emph{total} $\cETR$ search problems that are
  SL-reducible to a basic $\FIXP$ problem for which each domain $D_I$ is
  a convex polytope described by a set of linear inequalities with
  rational coefficients and the function $F_I$ is defined by a
  $\BasisPlusMinusTimesDivMaxMin$-circuit $C_I$.
\end{definition}
The class $\FIXPA$ is the class of strong approximation problems
corresponding to $\FIXP$. More precisely, $\FIXPA$ consist of all
discrete search problems polynomial time reducible to the problem
$\Pi_a$ for $\Pi \in \FIXP$.

The definition of $\FIXP$ is quite robust with respect to the choice
of domain and set of basis functions allowed by circuits in the basic
$\FIXP$ problems. Etessami and Yannakis proved that basic $\FIXP$
problems defined by $\BasisPlusMinusTimesDivMaxMinRoots$-circuits are
still in the class $\FIXP$. Likewise, basic $\FIXP$-problems where
$D_I$ is a ball with rational-valued center and diameter, or more
generally an ellipsoid given by a rational center-point and a
positive-definite matrix with rational entries, are still in the class
$\FIXP$~\cite[Lemma~4.1]{SICOMP:EtessamiY10}. The same argument allows
for using as domain the ball $B^d_p$ with respect to the $\ell_p$ norm
for any rational $p\geq 1$ or $p=\infty$, with the coordinates
possibly transformed by individual affine functions.

On the other hand, Etessami and Yannakakis also proved that one may
greatly restrict the class of basic $\FIXP$ problems used to define
$\FIXP$ without changing the class. The domains may be restricted to
be unit hypercubes $[0,1]^{d_I}$ and the circuits may be restricted to
$\BasisPlusTimesMax$-circuits. Both restrictions may in fact be
imposed at the same time. The restriction to
$\BasisPlusTimesMax$-circuits is a consequence of first proving
that the problem of finding a Nash equilibrium in a given finite game
in strategic form is hard for $\FIXP$ with respect to SL-reductions
and then proving $\FIXP$-membership of this problem using
$\BasisPlusTimesMax$-circuits.

Another way to restrict circuits is by limiting their depth. The
function of Nash for expressing Nash equilibrium as Brouwer fixed
points involve divisions but as noted by Etessami and Yannakakis it
may be viewed as a constant depth circuit, if one allows for addition
gates of arbitrary fanin. Thus in the definition of $\FIXP$ one may
restrict circuits to be constant depth $\BasisPlusTimesMax$-circuits,
where the addition gates are allowed to have unbounded fanin.

We show in Proposition~\ref{PROP:DepthOneFIXP} of
Section~\ref{SEC:Structural-FIXP} that one may in fact take this much
further and \emph{completely} flatten the circuits of defining
problems for $\FIXP$ to be depth~$1$ circuits of fanin at most~2,
additionally also without requiring division.  In other words, each
coordinate function becomes just a simple function of at most~2
coordinates of the input. We also show in
Proposition~\ref{PROP:CircuitReductionFIXP} that $\FIXP$ is closed
under \emph{much} more powerful reductions than just the basic
SL-reductions used to define the class $\FIXP$.

\subsubsection{The Borsuk-Ulam theorem and \BU}
A new class $\BU$ of total $\cETR$ search problems based on the Borsuk
Ulam theorem was recently introduced by
Deligkas~et~al.~\cite{JCSS:DeligkasFMS21}. The definition of $\BU$ is
meant to capture the Borsuk-Ulam theorem as stated in formulation~(1)
of Theorem~\ref{THM:Borsuk-Ulam}. Following the definition of $\FIXP$
by Etessami and Yannakakis, Deligkas~et~al.\ first consider a set of
basic search problems and then close the class under
reductions. 
\begin{definition}
\label{DEF:BasicBUgeneral}
  An $\cETR$ search problem $\Pi$ is a basic $\BU$ problem if every
  instance $I$ describes a domain $D_I \subseteq \RR^{d_I}$ which is
  homeomorphic to $S^{d_I-1}$ by an an antipode preserving homeomorphism
  and a continuous function $F_I \colon D_I \rightarrow \RR^{d_I-1}$,
  computed by an algebraic circuit $C_I$, and these descriptions must
  be computable in polynomial time. The solution set is
  $\Sol(I) = \{x\in D_I \mid F_I(x)=F_I(-x)\}$.
\end{definition}
In defining the class Deligkas~et~al.\ restrict their attention to
spheres with respect to the $\ell_1$-norm as domains and functions
computed by $\BasisPlusMinusTimesMaxMin$-circuits. Compared to the
definition of $\FIXP$, division gates are thus excluded. However we
show later in Section~\ref{SEC:BU-BBU} that division gates can always
be eliminated. Having thus fixed the set of basic $\BU$ search
problems what remains in order to define $\BU$ is to settle on a
notion of reductions. In their journal paper,
Deligkas~et~al.~\cite{JCSS:DeligkasFMS21} suggest using reductions
computable by general algebraic circuits including non-continuous
comparison gates, whereas in the preceeding conference
paper~\cite{ICALP:DeligkasFMS19} they did not precisely define a
choice of reductions. We shall revisit the question of choice of
reduction in Section~\ref{SEC:BU-BBU} before proposing our definition
of $\BU$.

\subsection{Consensus Halving}
We give here a formal definition of consensus halving with additive measures as real valued search problems.
\begin{definition}
\label{DEF:CH}
  The problem $\CH$ is defined as follows. An instance $I$ consists of
  a list  of $\BasisPlusMinusTimesDivMaxMin$-circuits $C_1,\dots,C_n$
  computing distribution functions $F_1,\dots,F_n$ defined on the
  interval $A=[0,1]$. The domain is $D_I=S^n_1$ and $\Sol(I)$
  constists of all $x$ for which
  \begin{equation}
    \label{EQ:CH-equation}
    \sum_{j : x_j>0} F_i(t_j)-F_i(t_{j-1}) = \sum_{j : x_j<0} F_i(t_j)-F(t_{j-1}) \enspace ,
  \end{equation}
where $t_0=0$ and $t_j=\sum_{k \leq j} \abs{x_k}$, for $j=1,\dots,n+1$.
\end{definition}
Given $\BasisPlusMinusTimesDivMaxMin$-circuits computing the
distribution functions $F_i$, the function $F$ computing the
left-hand-side of equation~(\ref{EQ:CH-equation}) may clearly be
computed by $\BasisPlusMinusTimesDivMaxMin$-circuits as well. The
result of Deligkas~et~al.\ that $\CH$ is contained in $\BU$ follows.

The existence proof of a consensus halving by Simmons and Su as well
the formulation of a $\cETR$ search problem by Deligkas~et~al.\ match
the Borsuk-Ulam theorem as stated in formulation~(1) of
Theorem~\ref{THM:Borsuk-Ulam}. We shall also define a variation
$\BCH$ of $\CH$ to match formulation~(3) of
Theorem~\ref{THM:Borsuk-Ulam}. A point $y \in B^n_1$ may be lifted
to the point $x=(1-\norm{y}_1,y) \in S^n_1$. This means that we may
view $y \in B^n_1$ as describing a partition of $A$ by the partition
described by $x$. Compared to the representation of partitions of $A$
into $n+1$ intervals given by points of $S^n_1$ we thus restrict the
label of the first interval to be~$+$, in case it has positive length.
\begin{definition}
\label{DEF:BCH}
  The problem $\BCH$ is defined as follows. An instance $I$ consists of
  a list of $\BasisPlusMinusTimesMaxMin$-circuits $C_1,\dots,C_n$
  computing distribution functions $F_1,\dots,F_n$ defined on the
  interval $A=[0,1]$. The domain is $D_I=B^n_1$ and $\Sol(I)$
  constists of all $y$ for which
  \begin{equation}
    \label{EQ:BCH-equation}
    F_i(t_1) + \sum_{j : y_j>0} F_i(t_{j+1})-F_i(t_j) = \sum_{j : y_j<0} F_i(t_{j+1})-F(t_j) \enspace ,
  \end{equation}
  where $t_0=0$, and $t_j=1-\sum_{k\geq j} \abs{y_j}$, for
  $j=1,\dots,n+1$.
\end{definition}

\subsection{Tools from Real Algebraic Geometry}
\label{SEC:ToolsRAG}

For obtaining our results concerning strong approximation we need
concrete bounds on $\delta>0$ as a function of $\eps>0$ witnessing the
truth of ``epsilon-delta'' statements. When such a statement is
expressible in the first-order theory of the reals, such bounds can be
obtained in a generic way using the general machinery of real
algebraic geometry~\cite{Book:BasuPollackRoy-ed2-posted3-2016}. This approach has
been used several times previously for establishing $\FIXPA$
membership of the problem of strong approximation of Nash equilibrium
refinements~\cite{SAGT:EtessamiHMS2014,GEB:Etessami2020,EC:HansenL2018}.

Concretely, suppose that $\Phi(\eps,\delta)$ is a formula with free
variables $\eps$ and $\delta$ of the form
\[
  \Phi(\eps,\delta) = (Q_1 x_1 \in \RR^{n_1}) \cdots (Q_\omega x_\omega \in \RR^{k_\omega}) F(x_1,\dots,x_\omega,\eps,\delta) \enspace ,
\]
where $Q_i\in \{\forall,\exists\}$, and $F$ is a Boolean formula whose
atoms are polynomial equalities and inequalities involving polynomials
of degree at most $d$ and having integer coefficients of bitsize at
most~$\tau$.

We now assume that the statement
$(\forall \eps>0)(\exists \delta>0) \Phi(\eps,\delta)$ is true, and
fix $\eps=2^{-k}$, for a positive integer $k$, resulting in the
formula $(\delta>0) \wedge \Phi(\eps,\delta)$, with $\delta$ as the
only variable. We may now perform \emph{quantifier
  elimination}~\cite[Algorithm~14.21]{Book:BasuPollackRoy-ed2-posted3-2016}
on this to obtain an equivalent quantifier free formula
$\Psi(\delta)$. The formula $\Psi(\delta)$ is simply a Boolean formula
whose atoms involve univariate polynomial equalities and
inequalities. The bounds given by Basu, Pollack and Roy for the result
of quantifier elimination imply that the degree of the univariate
polynomials are bounded by $d^{O(k_1)\dots O(k_\omega)}$ with
coefficients of bitsize at most
$\max(k,\tau)d^{O(k_1)\dots O(k_\omega)}$. We may now appeal to
Theorem~13.17 of~\cite{Book:BasuPollackRoy-ed2-posted3-2016} to
conclude that $\Psi(\delta)$, and hence also $\Phi(\eps,\delta)$ is
true, for some
$\delta \geq 2^{-\max(k,\tau)d^{O(k_1)\dots O(k_\omega)}}$.

In our applications, the formula $\Phi$ is defined from a given
instance $I$. Both $\tau$ and $d$ will be bounded by fixed polynomials
of $\abs{I}$. The number of blocks $\omega$ of quantified variables
will be a fixed constant, and $k_i$ for $1\leq i\leq \omega$ are
bounded by fixed polynomials of $\abs{I}$ as well. In other words
there will be a fixed polynomial $q$ such that the formula
$\Phi(\eps,\delta)$ is true for some
$\delta \geq (1/\eps)^{2^{g(\abs{I})}}$.

The first-order formulas we consider are expressed using also the
evaluation of functions computable by algebraic circuits as a
primitive. We may in a generic way transform such formulas to having
only polynomial inequalities and equalities and required
above. Namely, we may perform a Tseitin-style transformation by
introducing existentially quantified variables for each gate of the
circuit and express using polynomial inequalities and equalities that
each gate is computed correctly, and the variables corresponding to
the output gates may then be used instead in place of the function. As
long as the number of evaluations of functions is constant, this
leaves the number of blocks of quantified variables constant.

\section{Structural Properties of \FIXP}
\label{SEC:Structural-FIXP}
Recall that $\FIXP$ is defined to be the closure of all basic $\FIXP$
problems with respect to the very simple notion of SL-reductions. We
first show that $\FIXP$ is in fact closed under \emph{general} circuit
reductions.
\begin{proposition}
\label{PROP:CircuitReductionFIXP}
Suppose that $\Pi$ is a $\cETR$ search problem defined with unit
hypercube domains and reduces to $\Gamma \in \FIXP$ by a polynomial
time $\BasisPlusMinusTimesDivMaxMinRoots$-circuit reduction. Then
$\Pi$ belongs to $\FIXP$ as well.
\end{proposition}
\begin{proof}
  We may without loss of generality assume the domain of $\Gamma$ is
  also the unit hypercube. Let $(f,g)$ be the assumed reduction from
  $\Pi$ to $\Gamma$. Let $I$ be an instance of $\Pi$. By assumption
  $D_I = [0,1]^m$ and $D_{f(I)} = [0,1]^n$, where $m=d_I$ and
  $n=d_{f(I)}$. From the definition of $(f,g)$ we may given $I$ in
  polynomial time compute $f(I)$ as well as the circuit $C_I$ that
  defines a function $G \colon [0,1]^n \rightarrow [0,1]^m$ such that
  $g(I,x)=G(x)$ for all $x \in \Sol(f(I))$. By assumption on $\Gamma$
  we may in polynomial time compute another circuit $C_{f(I)}$ that
  defines a function $F \colon [0,1]^n \rightarrow [0,1]^n$ such that
  $\Sol(f(I))$ are the fixed points of $F$.

  We now define the function $H \colon [0,1]^{n+m} \rightarrow [0,1]^{n+m}$
  by $H(x,y) = (F(x),G(x))$. Clearly the set of fixed points of $H$ is
  equal to $\{(x,G(x)) \mid x \in \Sol(f(I))\}$, and since $H$ is
  computable by a $\BasisPlusMinusTimesDivMaxMinRoots$-circuit this defines a $\cETR$
  search problem $\Lambda$ in $\FIXP$ with the same set of instances
  as $\Pi$. We note that the projection of a fixed point of $H$ to the
  last $m$ coordinates gives a solution to $\Pi$ from which it follows
  that $\Pi$ in particular SL-reduces to $\Lambda$. Therefore $\Pi$
  belongs to $\FIXP$ as well.
\end{proof}

Our next basic result is based on properties of the basic $\FIXP$ problem
used by Etessami and Yannakakis to show that the division operation is
not necessary to express all of $\FIXP$. We give a brief review of
their construction. An instance $I$ describes a $d$-player game in
strategic form. Player~$i$ has a set $S_i$ of $n_i=\abs{S_i}$ pure
strategies and a utility function
$u_i \colon S_1 \times \dots \times S_d \rightarrow \RR$. Let
$n=n_1+\dots+n_d$ be the total number of strategies. The domain is
given as
$D_I = \Simplex_{n_1-1} \times \dots \times \Simplex_{n_d-1}$, where
the $(n_i-1)$-dimensional unit simplex $\Simplex_{n_i-1}$ is
identified with the set of probability distributions on $S_i$, for
$i=1,\dots,d$. The domain $D_I$ may be viewed as a subset of $\RR^n$
in the natural way. The utility functions define the function
$v \colon D_I \rightarrow \RR^n$ given by
\[
  v(x)_{ia_i} = \sum_{a_{-i} \in S_{-i}} u_i(a_1,\dots,a_d)
  \prod_{j\neq i} x_{ja_j} \enspace ,
\]
where
$S_{-i} = S_1 \times \dots \times S_{i-1} \times S_{i+1} \times \dots
\times S_d$ and
$a_{-i} = (a_1,\dots,a_{i-1},a_{i+1},\dots,a_d) \in S_{-i}$. Define
further the function $h \colon D_I \rightarrow \RR^n$ by $h(x)=x+v(x)$ and
finally let $G_I \colon D_I \rightarrow D_I$ be defined by letting $G_I(x)$
be the projection of $h(x)$ onto $D_I$. For all $i=1,\dots,d$, it
holds that $G_I(x)_{i,a_i} = \max(h_{i,a_i}-t_i,0)$, where $t_i$ is
the unique value satisfying
$\sum_{a_i \in S_i} \max(h_{i,a_i}-t_i,0) = 1$. The fixed points of
$G_I$ are exactly the Nash equilibria of the game described by
$I$~\cite[Lemma~4.5]{SICOMP:EtessamiY10}, and the search problem is
therefore $\FIXP$-complete~\cite[Theorem~4.3]{SICOMP:EtessamiY10}.

The definitions of the functions $v$, $h$, and $G_I$ allows us to
extend their domain from $D_I$ to the $n$-dimensional unit cube
$[0,1]^n$. By definition of $G_I$ this does not change the set of
fixed points of $G_I$. Likewise, applying the same affine
transformation to $u_i(x)$, for $i=1,\dots,d$, does not change the set
of fixed points of $G_I$. We may thus assume that $u_i$ has codomain
$[0,1]$. Making use of a sorting network, Etessami and Yannakakis show
that $G_I$ may be computed by a polynomial size
$\BasisPlusMinusTimesMaxMin$-circuit
$C_I$~\cite[Lemma~4.6]{SICOMP:EtessamiY10}. It is furthermore
straightforward to ensure that all constants used in $C_I$ as well as
values computed by gate functions of $C_I$ belong to the interval
$[0,1]$ for any input $x \in [0,1]^n$ (cf.\
\cite{JCSS:DeligkasFMS21}). We summarize these observations below.
\begin{proposition}
\label{PROP:BasicFIXP-NE}
There is a basic $\FIXP$ problem $\Pi_\NE$, complete for $\FIXP$ under
SL-reductions, such that for any instance $I$ it holds that
$D_I=[0,1]^{d_I}$ and such that $C_I$ is a
$\BasisPlusMinusTimesMaxMin$-circuit that satisfies that all gate
functions of $C_I$ compute values in $[0,1]$ given input
$x \in D_I$.
\end{proposition}

From here we may derive a characterization of $\FIXP$ in terms of
depth~1 circuits, where the addition and subtraction operators
(necessarily) are truncated to the interval $[0,1]$. This is simply
done by a Tseitin-style transformation. One may note that a
Tseitin-style transformation is already used in the proof that
$\Pi_{\rm NE}$ is $\FIXP$-hard. This means such a transformation is
applied twice at different points of the proof to yield the statement
below.
\begin{proposition}
  \label{PROP:DepthOneFIXP}
  There is a basic $\FIXP$ problem $\Pi$, complete for $\FIXP$ under
  SL-reductions, such that for any instance $I$ it holds that
  $D_I=[0,1]^{d_I}$ and such that $C_I$ is a depth~1
  $\BasisTPlusTMinusTimesMaxMin$-circuit, using only constants from
  the interval $[0,1]$.
\end{proposition}
\begin{proof}
  We reduce from the problem $\Pi_{\mathrm NE}$ of
  Proposition~\ref{PROP:BasicFIXP-NE}. The instances of $\Pi$ are the
  same instances of $\Pi_\NE$. Let $I$ be an instance of $\Pi_\NE$ and
  let $D = [0,1]^{d_I}$ and $C_I$ be the corresponding domain and
  $\BasisPlusMinusTimesMaxMin$-circuit as given by
  Proposition~\ref{PROP:BasicFIXP-NE}.  Suppose that $C_I$ has $m_I$
  gates $g_1,\dots,g_{m_I}$. We define the new domain $D'_I$ for $\Pi$
  simply by $D'_I=[0,1]^{d'_I}$, where $d'_I=d_I+m_I$. We next define
  the gates of $C'_I$ which all are output gates of $C'_I$. We may
  consider the input as pairs $(x,y) \in [0,1]^{d_I}\times[0,1]^{m_I}$
  and we may think of the output gates as variables, similarly grouped
  as $(z,w)$ and ranging over $[0,1]^{d_I}\times[0,1]^{m_I}$. If $g_j$
  is an input gate labeled by $x_i$, we let $w_j=x_i$, and if $g_j$ is
  a constant gate labeled by $c \in [0,1]$ we let $w_j=c$. If $g_j$ is
  an addition gate taking as input gates $g_k$ and $g_\ell$ we let
  $w_j = (y_k + y_\ell)_{T[0,1]}$, i.e.\ the addition of $g_k$ and
  $g_\ell$ is simulated by a truncated addition of $y_k$ and
  $y_\ell$. The case of subtraction is analogous. If $g_j$ is a
  multiplication gate taking as input $g_k$ and $g_\ell$ we let
  $w_j = y_k\cdot y_\ell$. The case of maximum and minimum gates are
  analogous. Finally if $g_j$ is the $i$th output gate of $C_I$ we let
  $z_i=y_j$.  By construction $C'_I$ computes a function
  $F'_I \colon D'_I \rightarrow D'_I$ and $F'_I(x,y) = (x,y)$ if and only
  if $g_j$ computes the value $y_j$ on input $x$ for all $j$ and
  $C_I(x)=x$. We thus obtain $x$ such that $C_I(x)=x$ as the
  projection of $(x,y)$ to the first $d_I$ coordinates.
\end{proof}

In case we prefer to construct a normal
$\BasisPlusMinusTimesMaxMin$-circuit without truncated operations we
can clearly simulate the truncated addition and subtraction operations
by depth~$3$ circuits. We can also easily convert the circuits to
constant depth $\BasisPlusTimesMax$ circuits by considering the the
domain $B^{d'_I}_\infty=[-1,1]^{d'_I}$ instead of $[0,1]^{d'_I}$.

\section{Definition and Structural Properties of \BU\ and \BBU}
\label{SEC:BU-BBU}

In this section we define two classes of $\cETR$ search problems $\BU$
and $\BBU$ based on the Borsuk-Ulam theorem corresponding to formulations~(1)
and~(3) of Theorem~\ref{THM:Borsuk-Ulam}. We start by defining basic
$\BU$ and basic $\BBU$ problems. We shall restrict our attention to
the unit $n$-sphere and unit $n$-ball, but with regards to any
$\ell_p$ norm for $p\geq 1$ or $p=\infty$. For the case of $\BU$ this
amounts to specializing Definition~\ref{DEF:BasicBUgeneral}.
\begin{definition} 
  A basic $\BU$ problems is a basic $\ell_p$-$\BU$ problem if for every
  instace $I$ we have $D_I=S^{d_I}_p$.
\end{definition}
Similarly we define the set of basic $\BBU$ problems with respect to
the $\ell_p$-norm.
\begin{definition}
  An $\cETR$ search problem $\Pi$ is a basic $\ell_p$-$\BBU$ problem
  if for every instance $I$ we have $D_I=B^{d_I}_p$ and $I$ describes
  a continuous function $F_I \colon D_I \rightarrow \RR^{d_I}$, which
  is odd on the boundary $\boundary B^{d_I}_p$. The function $F_I$
  must be computed by an algebraic circuit $C_I$ whose description is
  computable in polynomial time. The solution set is
  $\Sol(I) = \{x\in D_I \mid F_I(x)=0\}$.
\end{definition}

The condition that the function $F_I$ is odd on $\boundary B^{d_I}_p$
is a semantic condition. However, typically the function $F_I$ would
be defined from a basic $\ell_p$-BU problem by a transformation done
in a similar way as in the proof of Theorem~\ref{THM:Borsuk-Ulam}, and
thereby $F_I$ would satisfy the condition automatically.

To define the classes $\BU$ and $\BBU$, we restrict our attention to
domains with respect to the $\ell_\infty$-norm.
\begin{definition}
  The class $\BU$ (respectively, $\BBU$) consists of all \emph{total}
  $\cETR$ search problems that are PL-reducible to a basic
  $\ell_\infty$-$\BU$ problem (respectively, basic
  $\ell_\infty$-$\BBU$ problem) for which the function $F_I$ is
  defined by a $\BasisPlusMinusTimesDivMaxMin$-circuit $C_I$.
\end{definition}

While the definition of $\BU$ in \cite{JCSS:DeligkasFMS21} was using
as domain the unit sphere with respect to the $\ell_1$-norm and not
allowing for division gates, we show in this section these changes do
not change the class. We propose choosing PL-reductions for closing
the class under reductions. PL-reductions are sufficient for obtaining
all of our results and they are polynomially continuous. Another
reason for this choice is that if we restrict the circuits defining
the classes $\FIXP$ and $\BU$ to also be piecewise linear, i.e.\ be
$\BasisPlusMulCMax$-circuits, we obtain the classes $\LinearFIXP$ and
$\LinearBU$, that when closed under polynomial-time reductions are
equal to $\PPAD$ and $\PPA$,
respectively~\cite{SICOMP:EtessamiY10,JCSS:DeligkasFMS21}.

\subsection{Elimination of Division Gates}
\label{SEC:EliminateDivision}

In this section, we show how to eliminate division gates from circuits
defining an instance of the $\BU$ or $\BBU$ problems. Let therefore $C$
denote an algebraic circuit defined over the basis 
$\BasisPlusMinusTimesDivMaxMinRoots$. 

\paragraph{Moving Divisions to the Top.}

In the paper ~\cite{SICOMP:EtessamiY10}, it is shown how to move all 
division gates to the top of the circuit by keeping track of the 
numerator and denominator of every gate. For sake of completeness we 
describe this transformation. Every gate $g_i$ is replaced by two gates 
$g_i'$ and $g_i''$ keeping track of the numerator and denominator, that 
is the value of $g_i$ in the original circuit will be equal to the 
value of $g_i' / g_i''$ in the transformed circuit. Firstly, if $g_i$ is 
an input gate or a constant-$c$ gate we put $g_i'=x_j$ for appropriate~$j$
(respectively 
$g_i' = c$) and $g_i ''=1.$ In order to maintain the equality $g_i = 
g_i'/g_i''$, we proceed as follows: if $g_i = g_j\pm g_k$ is an 
addition/subtraction gate in the original circuit, then we update the 
numerator and denominator to $g_i'=g_j'\cdot g_k''\pm g_k'\cdot g_j''$ 
and $g_i''=g_j''\cdot g_k''$; if $g_i = g_j\cdot g_k$, then 
$g_i'=g_j'\cdot g_k'$ and $g_i'' = g_j''\cdot g_k''$; if $g_i = g_j\div g_k$,
then $g_i' = g_j'\cdot g_k''$ and $g_i'' = g_j''\cdot g_k'$. For root 
gates, we note that if $g_j=g_j'/g_j''$ is input to a $\sqrt[k]{}-$gate $g_i$ 
for $k$ even, then $g_j \geq 0$, from which it follows that $\sgn(g_j')=
\sgn(g_j'')$. With this in mind, we see that we may maintain the numerator 
and denominator of $g_i$ by putting $g_i' = \sqrt[k]{g_j'g_j''}$ and 
$g_i''=\sqrt[k]{g_j''\cdot g_j''}$. Finally, for the $\max$-gate we note 
that $\max(ca,cb)=c\max(a,b)$ for $c\geq 0$. Using this we see that if 
$g_i = \max(g_j,g_k)$, then we may maintain the numerator and denominator 
via the formulas $g_i' = \max(g_j'\cdot g_j''\cdot (g_k'')^2,g_k'\cdot 
g_k''\cdot (g_j'')^2)$ and $g_i '' = (g_j'')^2\cdot (g_k'')^2$. We note 
that all this can be done only blowing up the size of the circuit by a 
constant factor. In the aforementioned paper, the authors then have 
division gates at the top outputting $out_i = out_i'/out_i''$. However, for our 
application this is unnecessary and we may completely remove division gates.

\paragraph{Removing Division Gates for $\BBU$.} 

Suppose that $\Pi$ is a $\BBU$ problem. Let $I$ be an instance of $\Pi$ 
and denote by $C_I$ an algebraic circuit computing a continuous function 
$F_I\colon B^{d_I}\rightarrow\RR^{d_I}$ that is odd on $S^{d_I-1}$ such 
that $\Sol(I)=\{x\in B^{d_I}\mid F_{I}(x)=0\}$. As described above, we may 
transform the circuit $C_I$ to a circuit $C_I^{+}$ that maintains the 
numerator and denominator of every gate. In the same way we define a circuit 
$C_I^{-}$ that is exactly like $C_I^{+}$, except it multiplies the input 
by $-1$ at the very beginning. Let $out_i^{n+},out_i^{d+}$ and 
$out_i^{n-},out_i^{d-}$ denote the gates in $C_{I}^{\pm}$ representing the 
numerators and denominators of the output gates of $C_I$. We now define a 
circuit $C_I^{*}$ that on input $x$ feeds this into $C_I^{+}$ and $C_I^{-}$ 
and then outputs the values $out_i^{n+}\cdot out_i^{d-}$ for $i=1,\dots, d_I$. 
If we denote by $F_i=F_i'/F_i''$ the coordinate functions of $F_I$, then 
$C_I^{*}$ is a circuit computing the function $F_I^{*}$ with coordinate 
functions $F_i'(x)F_i''(-x).$ Now, if $x\in S^{d_I-1}$ then 
$F_i'(x)/F_i''(x)=-F_i'(-x)/F_i''(-x),$ so $F_i'(x)F_i''(-x)=-F_i'(-x)
F_i''(-(-x))$, meaning that $F_I^{*}$ is odd on the boundary. In this way 
we have defined a $\BBU$ problem $\Gamma$ with the same instances as $\Pi$. 
Furthermore, given an instance $I$ of $\Pi$ one may in polynomial time compute 
an instance $f(I)$ of $\Gamma$ by computing $C_I^{*}$. We note that for any 
$x\in B^{d_I}$ it holds that $F_I(x)=0$ if and only if $F_I^{*}(x)=0$. We 
conclude that $\Pi$ SL-reduces to the division-free $\BBU -$problem $\Gamma$.

\paragraph{Removing Division Gates for $\BU$.} 

Now let $I$ be an instance of a $\BU -$problem $\Pi$ and denote by $C_I$
an algebraic circuit computing a continuous function $F_I\colon S^{d_I}
\rightarrow\RR^{d_I}$ such that $\Sol(I) = \{x\in S^{d_I}\mid F_I(x)=F_I(-x)\}$.
We make the same reduction as for $\BBU$ defining a circuit $C_I^{*}$ that
computes a function $F_I^{*}\colon S^{d_I}\rightarrow \RR^{d_I}$ whose 
coordinate functions are given by $F_i'(x)F_i''(-x)$ where $F_i'(x)/F_i''(x)$ is the 
$i$th coordinate function of $F_I$. By definition, $x$ is a BU-point of $F_I$ 
if and only if $F_i'(x)/F_i''(x)=F_i'(-x)/F_i''(-x)$ for all $i$. This happens
if and only if $F_i'(x)F_i''(-x)=F_i'(-x)F_i''(-(-x))$ for all $i$, meaning that
$x$ is a BU-point of $F_I^{*}$. Again, we conclude that $\Pi$ SL-reduces to a 
division-free $\BU -$problem. 


In the previous two paragraphs, we have shown the following result. 

\begin{proposition}
The classes $\BU$ and $\BBU$ remain the same even if the circuits 
are restricted to not using division gates. 
\end{proposition}

\subsection{Relationship with \FIXP}
As a consequence of their results about consensus halving,
Deligkas~et~al.\ proved that $\FIXP \subseteq \BU$. We observe here
that the direct proof that the Bosuk-Ulam theorem implies the Brouwer
fixed point theorem due to Volovikov~\cite{AMM:Volovikov08} gives a
much simpler way to derive this relationship. For completeness we
present the construction and proof of Volovikov.
\begin{proposition}[Volovikov]
  \label{PROP:Volovikov}
  Let $f \colon B^d_\infty \rightarrow B^d_\infty$ be a continuous
  function. Define the continous function
  $g \colon S^d_\infty \rightarrow \RR^d$ by
  $g(x,t) = (1+t)(tf(x)-x)$. If $g(x,t)=g(-x,-t)$ then $\abs{t}=1$ and $f(tx)=tx$.
\end{proposition}
\begin{proof}
  Note first that
  \[
    g(x,t)-g(-x,-t)=t\left[(1+t)f(x)+(1-t)f(-x)\right]-2x \enspace .
  \]
  It follows that $g(x,t)=g(-x,-t)$ if and only if $k(x,t)=x$, where
  \[
    k(x,t)=\frac{t}{2}\left[(1+t)f(x)+(1-t)f(-x)\right] \enspace .
  \]
  If $(x,t) \in S^d_\infty$ and $\abs{t}<1$ it holds that
  $\norm{x}_\infty=1$. Then since
  \[
    \begin{split}
    \norm{k(x,t)}_\infty \leq & \frac{\abs{t}}{2}\left[(1+t)\norm{f(x)}_\infty + (1-t)\norm{f(-x)}_\infty\right] \\ \leq & \frac{\abs{t}}{2}\left[(1+t) + (1-t)\right] = \abs{t} < 1 \enspace ,
  \end{split}
\]
we have $k(x,t)\neq x$. Thus $g(x,t)=g(-x,-t)$ implies that
$\abs{t}=1$. When $\abs{t}=1$ we clearly have $k(x,t)=tf(tx)$. In
conclusion, $g(x,t)=g(-x,-t)$ implies $t f(tx)=x$, or equivalently
that $f(tx)=tx$.
\end{proof}

The above construction immediately give a simple reduction from any
basic $\FIXP$ problem with domains $B^{d_I}_\infty$ to a basic
$\ell_\infty$-$\BU$ problem. The solution mapping of the reduction
must map solutions $(x,t)$ to $tx$. This may be done by simply using
multiplication gates. But since any solution $(x,t)$ has $\abs{t}=1$
the multiplication $tx_i$ may also be expressed as
$\Sel_2(-x_i,x_i,t)$, which means the solution mapping can also be
computed by constant depth $\BasisPlusTimesMax$-circuits.
\begin{proposition}
  Any $\Pi \in \FIXP$ reduces to a basic $\ell_\infty$-$\BU$ problem
  with $\BasisPlusMinusTimesMaxMin$-circuit by polynomial time constant depth
  $B$-circuit reductions, for both $B=\BasisPlusMinusTimes$ and
  $B=\BasisPlusMinusMulCMaxMin$.
\end{proposition}
\begin{proof}
  Any $\Pi \in \FIXP$ SL-reduces to a basic $\FIXP$ problem $\Gamma$
  with domains $D_I=B^{d_I}_\infty$ and $\BasisPlusTimesMax$-circuits
  $C_I$. From that, the instance mapping as described by
  Proposition~\ref{PROP:Volovikov} produces a
  $\BasisPlusTimesMax$-circuit and domain $S^{d_I}_\infty$. The
  composition of the SL-reduction and the reduction described above
  then yields the claimed types of reductions.
\end{proof}

\subsection{Change of Domains for $\BU$ and $\BBU$}
In this section we show reduce between different domains for the BBU
and BU problems.

\begin{proposition}
Let $B$ be a set of gates that contains $\{+,-,\ast,\div,\max,\min\}$.
Suppose that $\Pi$ is an $\exists\RR$ search problem whose domains are contained in hypercubes
that reduces to a basic $\ell_p - \BBU$ problem $\Gamma$ by a polynomial time $B$-circuit 
reduction $(f,g)$. Furthermore, 
suppose that for any instance $I$ of $\Pi$ the function $g(I,\cdot)$ mapping solutions
of $f(I)$ to solutions of $I$ is odd and assume that $C_{f(I)}$ is also a $B-$circuit.
\textit{(i)} If $p=\infty$ then $\Pi$ SL-reduces to a basic $\ell_{\infty}-\BBU$ problem 
using gates in $B$. \textit{(ii)} If $1\leq p<\infty$  then $\Pi$ SL-reduces to a basic
$\ell_p -\BBU$ problem using $B\cup \{\sqrt[p]{\cdot}\}-$circuits.
\label{PROP:Search-to-BBU}
\end{proposition}
\begin{proof}
\textit{(i)} First assume that the domains of $\Gamma$ are unit hypercubes.
Let $I$ denote an instance  of $\Pi$. By assumption $D_I \subseteq [-1,1]^{m}$ and 
$D_{f(I)} = [-1,1]^{n}$ where $m=d_I$ and $n=d_{f(I)}$. From the definition 
of $(f,g)$ we may given $I$ in polynomial time compute $f(I)$ and a circuit 
$C_I$ computing a function $G\colon [-1,1]^n\rightarrow [-1,1]^m$ such  that 
$G(x)=g(I,x)\in\Sol(I)$ for every $x\in\Sol(f(I)).$ By assumption of $\Gamma$
 we may in polynomial time compute another circuit $C_{f(I)}$ that defines a 
function $F\colon [-1,1]^n\rightarrow\RR^n$ that is odd on the boundary such 
that $\Sol(f(I))$ are the zeroes of $F$. 

Define $H \colon [-1,1]^{n+m}\rightarrow\RR^{n+m}$ by $H(x,y)=
((1-||y||_{\infty})F(x),y- \tfrac{1}{2}G(x)))$. As $G$ is odd and $F$ is odd 
on the boundary, one may verify that $H$ is odd on the boundary of  
$[-1,1]^{m+n}$. As $H$ is polynomial-time computable by a $B-$circuit,
it defines an $\ell_{\infty}- \BBU$ problem $\Lambda$ with the same 
instances as $\Pi$. Furthermore, if $(x,y)$ is a zero of $H$, then $y=G(x)/2$,
so $||y||_{\infty}<1$. The equality $(1-||y||_{\infty})F(x)=0$ from the first
component then implies $F(x)=0$. Therefore, the zeroes of $H$ are
contained in $\{(x,G(x)/2)\mid x\in\Sol(f(I))\}$.  Given a zero of $H$ one may 
recover a solution to $\Pi$ by projecting onto the last $m$ coordinates 
and multiplying by $2$. In particular, $\Pi$ SL-reduces to $\Lambda$. 


\textit{(ii)} Now, suppose that the domains of $\Gamma$ are $p$-balls,
 where $1\leq p<\infty$. Again by assumption we have that $D_I \subseteq 
[-1,1]^m$ and $D_{f(I)}=B_{p}^n$ where $m=d_I$ and $n=d_{f(I)}$, 
and we may given an instance $I$ of $\Pi$ in polynomial time compute a 
circuit $C_I$ defining a function $G\colon B_p^n \rightarrow [-1,1]^m$ such 
that $G(x)=g(I,x)\in\Sol(I)$ for every $x\in\Sol(f(I))$. Furthermore, we 
may in polynomial time compute a circuit $C_{f(I)}$ computing a function 
$F\colon B_{p}^{n}\rightarrow\RR^n$ that is odd on $S_p^{n-1}$ such 
that $\Sol(f(I))$ is the zeroes of $F$. 

Now define an odd function $h\colon B_p^n\rightarrow B_p^n$ by 
$h(x)=x/\max(1/2,||x||_p)$, which may be computed by a circuit
using also $\sqrt[p]{\cdot }$ gates, and define $H\colon B_p^{n+m}
\rightarrow\RR^{n+m}$ by
\begin{align*}
H(x,y)=(\max(0,\tfrac{1}{2}-||y||_p^p)F(h(x)),y-\tfrac{1}{3n}G(h(x)))
\end{align*}
First we remark that $H$ is odd on the boundary of
 $B_p^{n+m}$. Clearly, the second coordinate is always odd, and the 
first coordinate evaluates to $0$ if  $||y||_p^p>1/2$. If $(x,y)\in S_p^{n+m-1}$
 and $||y||_p^p<1/2$, then $||x||_p^p>1/2$ which implies that $||x||_p > 1/2$.
 This then implies that $h(x)=x/||x||_p$ and so $F(h(x))=-F(-h(x))=-F(h(-x))$,
 because $h$ is odd and  $F$ is odd on $S_p^{n-1}$.

Now, if $(x,y)$ is a zero 
of $H$, then $y=\tfrac{1}{3n}G(h(x))$ and so $||y||_p^p\leq 
(n||y||_{\infty})^p\leq \tfrac{1}{3^p}<\tfrac{1}{2}$. From the first first coordinate equality 
$\max(0,1/2-||y||_p^p)F(h(x))=0$ one then obtains that $F(h(x))=0$ so
$h(x)\in \Sol(f(I))$. Thus, the set of zeroes of $H$ are contained in $\{(x,\frac{1}{3n}G(h(x))
\mid h(x)\in\Sol(f(I))\}$. Furthermore, $H$ can be computed by circuit over 
$B\cup\{\sqrt[p]{\cdot }\}$, so this defines a basic $\ell_p -\BBU$ problem 
$\Lambda$ with $(B\cup\{\sqrt[p]{\cdot }\})$-circuits and the same instances as $\Pi$. From
a zero of $H$ we may again recover a solution to $\Pi$ by projecting onto 
the last $m$ coordinates and multiplying the result by $3n$. We conclude that 
$\Pi$ SL-reduces to $\Lambda.$
\end{proof}

\begin{proposition}
Any basic $\ell_{\infty} -\BBU$ problem SL-reduces to a basic $\ell_p 
-\BBU$ problem using gates in $\{+,-,\ast,\div,\max,\min,\sqrt[p]{\cdot}\}$.
\label{PROP:BBU-inf-to-p}
\end{proposition}
\begin{proof}
Let $\Pi$ be a basic $\ell_{\infty} -\BBU$ problem. By the previous proposition it 
suffices to argue that $\Pi$ polynomial time $\{+,-,\ast,\div,\max,\min\}
-$reduces to a a basic $\ell_p -\BBU$ problem. Given an instance
$I$ of $\Pi$, compute in polynomial time a circuit $C_I$ defining a
function $F\colon B_{\infty}^{n}\rightarrow\RR^{n}$ that is odd on $S_p ^{n-1}$ such that 
$\Sol(I)$ are the zeroes of $F$. Also, define the map $\pi\colon B_{\infty}^n\rightarrow
B_{\infty}^n$ by $\pi(x) = x/\max(\tfrac{1}{2n},||x||_{\infty})$. Now we may in polynomial 
time compute $\{+,-,\ast,\div,\max,\min\}-$circuit computing the function 
$G\colon B_p ^n\rightarrow \RR^n$ given by $G(x)=F(\pi(x))$. If $||x||_p = 1$, 
then $||x||_{\infty}\geq 1/(2n)$, so $||\pi(x)||_{\infty} = 1$, implying that $G(x)=G(-x)$. 
Thus we have defined a map $f$ taking instances $I$ of $\Pi$ to instances $f(I)$ of a 
basic $\ell_p-\BBU$ problem $\Gamma$. We note that $f$ is computable in polynomial time.
Furthermore, from $x\in\Sol(f(I))$ one may recover a solution by $g(I,x)=\pi(x)$ to $I$. 
As the function $g(I,\cdot)$ is odd and computable by a $\{+,-,\ast,\div,\max,\min\}-$circuit 
we conclude that $(f,g)$ satisfies the requirements of the previous proposition. We conclude 
that  $\Pi$ SL-reduces to a $\ell_p -\BBU$ using gates from $\{+,-,\ast,\div,\max,
\min,\sqrt[p]{\cdot}\}$. 
\end{proof}

\begin{proposition}
Any basic $\ell_p -\BBU$ problem $\Pi$ SL-reduces to a basic $\ell_{\infty} 
-\BBU$ problem where the circuits are allowed to use $\sqrt[p]{\cdot}$ gates.
\label{PROP:BBU-p-to-inf}
\end{proposition}
\begin{proof}
Let $\Pi$ be a basic $\ell_p-\BBU$ problem and let $I$ denote an instance of $\Pi$.
We may compute a circuit $C_I$ defining a function $F\colon B_p^n\rightarrow\RR^n$
that is odd on the boundary of $B_p^n$ such that $\Sol(I)$ is the set of zeroes of $F$. 
Now, define a function $h\colon\RR^n\rightarrow\RR^n$ by $h(x) = x/\max(1/2,||x||_p)$. 
We may now in
polynomial time compute a $\{+,-,\ast,\div,\max,\max,\sqrt[p]{\cdot}\}$-circuit computing 
the function $H\colon B_{\infty}^n\rightarrow \RR^n$ given by $H(x)=F(h(x))$. 

If $x\in \partial B_{\infty}^n$, then $||x||_p\geq ||x||_{\infty} = 1$ which
shows that $h(x)=x/||x||_p$ so $||h(x)||_p = 1$. As $h$ is odd and $F$ is odd on 
the boundary, it follows that $H(x)= F(h(x))=F(-h(-x))=-F(h(-x))=-H(-x)$ showing
that $H$ is odd on the boundary of $B_{\infty}^n$, so it defines an instance of an
$\ell_{\infty}-\BBU$ problem $\Gamma$. Mapping back solutions amounts to computing
$h(x)$ which can be done by a $\{+,-,\ast,\div,\max,\max,\sqrt[p]{\cdot}\}$-circuit. 
The result now follows from part~(i) of Proposition~\ref{PROP:Search-to-BBU}.
\end{proof}

Now we proceed with showing the reductions between basic $\ell_p - \BU$
problems.

\begin{proposition}
Suppose that $\Pi$ is an $\exists\RR$ search problem whose domains are contained in 
hypercubes that reduces to a basic $\ell_p -\BU$ problem by a $B-$circuit reduction
$(f,g)$, where $\{+,-,\ast,\div,\max,\min\}\subseteq B$. Assume also that for
every instance $I$ of $\Pi$, and that $g(I,\cdot)$ is an odd map 
$\RR^{d_{f(I)}}\rightarrow [-1,1]^{d_I}$. \textit{(i)} If $p=\infty$
then $\Pi$ SL-reduces to a basic $\ell_{\infty}-\BU$ problem with circuits over $B$. 
\textit{(ii)} If $1\leq p<\infty$ then $\Pi$ SL-reduces to a basic $\ell_p-\BU$ 
problem with circuits over $B\cup\{\sqrt[p]{\cdot}\}$. 
\label{PROP:search-to-BU}
\end{proposition}
\begin{proof}
\textit{(i)} Let $I$ denote an instance of $\Pi$ and let $m=d_I$.
By assumption of $(f,g)$
we may in polynomial time compute a circuit defining a function $F\colon 
S_{\infty}^{n}\rightarrow\RR^{n}$ such that $\Sol(f(I))$ consists of the
$x\in S_{\infty}^{n}$ such that $F(x)=F(-x).$ By the result in Section~\ref{SEC:EliminateDivision}
we may assume that the circuit computing $F$ is division-free, and so we may
extend the domain of $F$ to be $\RR^{n+1}$. Also, we may in polynomial
time compute a circuit defining a function $G\colon \RR^{n+1}\rightarrow
[-1,1]^{m}$ mapping $\Sol(f(I))$ to $\Sol(I)$. Define $H\colon
S_{\infty}^{m+n}\rightarrow\RR^{m+n}$ by $H(x,y) = (F(x),y-\tfrac{1}{2}G(x))$.
We note that $H$ may be 
computed by a circuit over $B$, so it defines a basic $\ell_{\infty}-\BU$
problem $\Lambda$ with $B-$circuits and the same instances as $\Pi$. 
If $(x,y)\in S_{\infty}^{n+m}$ has $H(x,y)=H(-x,-y)$ then $y-\tfrac{1}{2}G(x)
=-y+\tfrac{1}{2}G(x)$, as $G$ is odd. Therefore $||y||_{\infty}=
||G(x)/2||_{\infty}<1$ and so $||x||_{\infty}=1$. Also, the first coordinate shows that 
$F(x)=F(-x).$ This says that $x\in\Sol(f(I))$, and so $G(x)\in \Sol(I)$. As
$2y=G(x)$ we see that $\Pi$ SL-reduces to $\Lambda.$


\textit{(ii)} Again let $I$ denote an instance of $\Pi$ with $m=d_I$. From $f(I)$ we may
in polynomial time compute a circuit computing a map $F\colon S_p^n
\rightarrow\RR^n$ such that $\Sol(f(I))=\{x\in S_p^n\mid F(x)=F(-x)\}$
and a $B-$circuit computing a map $G\colon \RR^{n+1}\rightarrow [-1,1]^{m}$
sending $\Sol(f(I))$ to $\Sol(I)$. Again, we may extend the domain of $F$.
 Define a map $h\colon\RR^{n+1}
\rightarrow \RR^{n+1}$ by $h(x) = x/\max(1/2,||x||_p)$ and $H\colon S_{p}^{m+n}
\rightarrow\RR^{m+n}$ by
\begin{align*}
H(x,y) = (F(h(x)), y-\tfrac{1}{2}G(h(x)))
\end{align*}
In this way, we have defined a basic $\ell_p -\BU$
problem $\Lambda$ with $B\cup \{\sqrt[p]{\cdot}\}-$circuits and the same 
instances as $\Pi$. If $(x,y)\in S_p^{n+m}$ has $H(x,y)=H(-x,-y)$ we find that 
$y = \tfrac{1}{2}G(h(x))$ so $||y||_p\leq 1/2$. This implies that $||x||_p \geq 1/2$,
and so $h(x)=x/||x||_p\in S_p^n$. Also, the first component shows that $F(h(x))=F(h(-x))
=F(-h(x))$ where we use that $h$ is odd. Therefore, $h(x) \in \Sol(f(I))$, 
and so $G(h(x))\in\Sol(I)$. As $2y=G(h(x))$,
we conclude that $\Pi$ SL-reduces to $\Lambda.$ 
\end{proof}

\begin{proposition}
Let $B =\BasisPlusMinusTimesDivMaxMin$. \textit{(i)} A basic $\ell_p -\BU$ 
problem $\Pi$ with $B$-circuits SL-reduces to a basic  $\ell_{\infty}-\BU$ 
problem with $B\cup\{\sqrt[p]{\cdot}\}$-circuits. \textit{(ii)} A basic 
$\ell_{\infty} -\BU$ problem $\Pi$ with $B$-circuits SL-reduces to a basic 
$\ell_{p}-\BU$ problem using $B$-circuits.
\label{PROP:BU-p-to-inf}
\end{proposition}
\begin{proof}
\textit{(i)} Let $\Pi$ denote a basic $\ell_p -\BU$ problem. Suppose an instance $I$
is defined by some continuous function $F\colon S_p^n\rightarrow\RR^n$. By Section~\ref{SEC:EliminateDivision} we may assume that the circuit computing $F$ is division-free and so extend
$F$ to be defined in all of $\RR^{n+1}$.  Define
a function $g\colon\RR^{n+1}\rightarrow [-1,1]^{n+1}$ by $g(x)=x/\max(1/2,||x||_p)$
and $H\colon S_{\infty}^n\rightarrow\RR^n$ by $H(x)=\overline{F}(g(x))$. Let $f$ denote
the map sending the instance $I$ to the instance $f(I)$ given by $H$ of a basic $\ell_{\infty}
-\BU$ problem $\Gamma.$ One may verify that $(f,g)$ is a reduction satisfying the properties
of Proposition~\ref{PROP:search-to-BU}, so by part~(i) of Proposition~\ref{PROP:search-to-BU} we have that $\Pi$ SL-reduces to the basic
$\ell_{\infty}-\BU$.


\textit{(ii)} Let $\Pi$ denote a basic $\ell_p -\BU$ problem. Suppose an instance $I$
is defined by some continuous function $F\colon S_{\infty}^n\rightarrow\RR^n$. Again,
we may extend $F$. Similarly to the case above, we define
a function $g\colon\RR^{n+1}\rightarrow [-1,1]^{n+1}$ by $g(x)=x/\max(1/(n+1),||x||_p)$
and $H\colon S_{p}^n\rightarrow\RR^n$ by $H(x)=\overline{F}(g(x))$. Let $f$ denote
the map sending the instance $I$ to the instance $f(I)$ given by $H$ of a basic $\ell_{\infty}
-\BU$ problem $\Gamma.$

First, the map $g$ satisfies the condition of Proposition~\ref{PROP:search-to-BU}. If $x\Sol(f(I))$ then it holds that
$x\in S_p^n$, and so $1=||x||_p\leq (n+1)||x||_{\infty}$, implying that $||x||_{\infty}\geq 1/(n+1)$.
From this it follows that $g(x)=x/||x||_{\infty}$ by definition. Furthermore, using that $g$ is odd
we find that $F(g(x))=H(x)=H(-x)=F(g(-x))=F(-g(x))$. We conclude that $g(x)\in\Sol(I)$. In
conclusion, $(f,g)$ is a reduction from $\Pi$ to $\Gamma$ satisfying the properties of Proposition~\ref{PROP:search-to-BU}.
By part~(ii) of Proposition~\ref{PROP:search-to-BU} we conclude that $\Pi$ SL-reduces to a basic $\ell_{\infty}-\BU$ problem.
\end{proof}

\section{Relation between $\ell_p -\BU$ and $\ell_p -\BBU$}

Let $B$ be some finite set of gates containing $\BasisPlusMinusTimesDivMaxMin$.
In this section we study reductions between $\ell_p -\BU$ problems and
$\ell_p -\BBU$ problems. Suppose we are given a basic $\ell_p -\BU$ problem $\Pi$ 
with circuits defined over $B$. In order to show that $\Pi$ reduces to a basic $\ell_p -
\BBU$ problem we follow the proof of Theorem~\ref{THM:Borsuk-Ulam}. Given an instance of $\Pi$ we may in 
polynomial time compute the dimension $n=d_I$ and a circuit over $B$ defining a map 
$F_I\colon S_p^n \rightarrow\RR^{n}$ such that $\Sol(I) = \{x\in S_p^n\mid F_I(x)=
F_I(-x)\}$. Define also the map $\pi\colon B_p ^n\rightarrow S_p^n$ by
\begin{align*}
\pi(x) = 
  \begin{cases}
    (x,(1-||x||_p^p)^{1/p}) & \text{ if } 1\leq p<\infty\\  (x/t, 2\cdot (1-t)) & \text{ if }p=\infty
  \end{cases}
\end{align*}
where $t=\max(1/2,||x||_{\infty})$. Define a map $H\colon B_p^n\rightarrow\RR^n$
by $H(x) = F_I(\pi(x))-F(-\pi(x))$. If $x\in S_p^{n-1}$ then the last coordinate of $\pi(x)$
vanishes and so $\pi(x)=-\pi(-x)$ implying that $H(x)=-H(-x)$, so $H$ is odd on the boundary.
As $H$ is computable by a $(B\cup\{\sqrt[p]{\cdot}\})-$circuit if $p<\infty$ (and $B-$circuit if
$p = \infty$) this defines an $\ell_p -\BBU$ problem $\Gamma$ with the same instances as
$\Pi$. Furthermore, the set of $\BU$-points of $H$ is exactly $\{x\in B_p^n\mid F_I(\pi(x))
= F_I(-\pi(x))\}$, so mapping solutions $x$ of $\Gamma$ to solutions of $\Pi$ amounts to 
computing  $\pi(x)$ which can be done by a circuit over $B\cup\{\sqrt[p]{\cdot}\}$ if $p<
\infty$ (and over $B$ if $p = \infty$). 

However, when $p\neq 1$ these reductions make use of $\sqrt[p]{}$-gates for $p<\infty$
or divison gates for $p = \infty$. We can remedy this by applying Propositions~\ref{PROP:BBU-inf-to-p},~\ref{PROP:BBU-p-to-inf},~and~\ref{PROP:BU-p-to-inf} which
give that we may go back and forth between different domains for $\BBU$ and $\BU$ by SL-reductions.
Specifically, for any $\ell_p -\BU$ problem we may SL-reduce to a $\ell_1-\BU$ problem (that also
uses $\sqrt[p]{}$ gates if $p<\infty$). Then we may apply the above $\{+,\MulC\}$-reduction
from $\ell_1 -\BU$ to $\ell_1 -\BBU$. And from there we may again SL-reduce to an $\ell_p -\BBU$ 
problem. In conclusion we obtain the following result.

\begin{proposition}
Any basic $\ell_p -\BU$ problem $\{+,\MulC\}$-reduces
to a basic $\ell_p -\BBU$ problem.
\end{proposition} 

Note that the reductions of this proposition are a special case of
PL-reductions. Because these are polynomially continuous, we
automatically also get the following result.

\begin{proposition}
Any basic $\ell_p -\BU_a$ problem polynomial time reduces
to a basic $\ell_p -\BBU_a$ problem.
\label{PROP:BasicBUaToBasicBBUa}
\end{proposition} 

For reductions in the other direction, consider an instance $H\colon B_{\infty}^n\rightarrow\RR^n$
of a basic $\ell_{\infty}-\BBU$ problem. Given this instance we define an instance of a basic 
$\ell_{\infty}-\BU$ problem given by $F\colon S_{\infty}^n\rightarrow\RR^n$ where
\begin{align*}
F(x) = \Sel_2(-H(-\pi(x)),H(\pi(x)),x_{n+1})
\end{align*}
and $\pi\colon\RR^{n+1}\rightarrow\RR^n$ is the projection $\pi(x_1,\dots, x_{n+1})
=(x_1,\dots, x_n)$. Now suppose that $x\in S_{\infty}^n$ satisfies $F(x)=F(-x)$. If
$x_{n+1}=1$ this implies that
\begin{align*}
H(\pi(x))=F(x)=F(-x)=-H(-\pi(-x))=-H(\pi(x))
\end{align*}
showing that $H(\pi(x))=0$, so $\pi(x)$ is a solution to the original problem. Similarly, if 
$x_{n+1}=-1$, then $H(-\pi(x))=0$, so $-\pi(x)$ is a solution to the original problem. In the case
where $|x_{n+1}|<1$ we have that $||\pi(x)||_{\infty}=1$ and so $H(\pi(x))=-H(-\pi(x))$,
because $H$ is odd on the boundary. By definition of the selection-function $\Sel$ this implies that
\begin{align*}
F(x)=\Sel_2(-H(-\pi(x)),H(\pi(x),x_{n+1})=\Sel_2(H(\pi(x)),H(\pi(x)),x_{n+1})=H(\pi(x))
\end{align*}
and similarly $F(-x)=-H(\pi(x))$. The equality $F(x)=F(-x)$ then implies that $H(\pi(x))
=H(-\pi(x))=0$, so both $\pi(x)$ and $-\pi(x)$ is a solution to the original instance in this case\footnote{We are grateful to Alexandros Hollender for noting that the reduction is possible without introducing approximation error.}. 
In conclusion, if we could recover the sign of $x_{n+1}$ then we could define a solution
map sending $x$ to $\sgn(x_{n+1})\pi(x)$, but we do not allow this. However, in the
approximate version, we may do this.

\begin{proposition}
Any basic $\ell_p-\BBUA$ problem polynomial time reduces to a basic basic $\ell_p -\BUA$ problem. 
\label{PROP:BasicBBUaToBasicBBUa}
\end{proposition}
\begin{proof}
After changing domain we may assume that $p=\infty$. Given an instance $(H,\epsilon)$ of a 
basic $\ell_{\infty}-\BBUA$ problem we apply the above construction and the map $f$ outputs
the instance $(F,\epsilon')$ of a basic $\ell_{\infty}-\BUA$ problem where $\epsilon' = \min
(\epsilon,1/2).$ Now suppose that $x$ is a solution to the problem $(F,\epsilon')$ . This means
there exists some $x^{*}$ with $||x-x^*||_{\infty}\leq\epsilon'$ and $F(x^*)=F(-x^*).$
We now claim that we may map back the solution $x$ of $(F,\epsilon')$ to a solution of $(H,\epsilon)$
by the map $g(x)=\sgn(x_{n+1})\pi(x)$. 

If $|x_{n+1}|\geq 1/2$ then we have that $\sgn(x_{n+1})=\sgn(x_{n+1}^{*})$ as $\epsilon'
\leq 1/2$. Therefore
\begin{align*}
||\sgn(x)\pi(x)-\sgn(x^{*})\pi(x^*)||_{\infty}=||\pi(x)-\pi(x^{*})||_{\infty}\leq ||x-x^{*}||\leq\epsilon'\leq\epsilon
\end{align*}
Also $\sgn(x^{*})\pi(x^{*})$ is a zero of $H$ by the discussion above the proposition. In the case
where $|x_{n+1}|< 1/2$ we have that $|x_{n+1}^*|<1$ and so both of $\pm\pi(x^{*})$ is a zero
of $H$. As $\sgn(x)\pi(x)$ is $\epsilon$-close to $-\pi(x^{*})$ or $\pi(x^{*})$, we conclude
that $g(x)$ is a solution to the problem $(H,\epsilon)$.  
\end{proof} 

Combining Proposition~\ref{PROP:BasicBUaToBasicBBUa} and Proposition~\ref{PROP:BasicBBUaToBasicBBUa} we obtain the following result.

\begin{theorem}
$\BUA=\BBUA$
\end{theorem}

\section{Consensus Halving}
\label{SEC:BBU-to-CH}
In this section we present the proof of our main result
Theorem~\ref{THM:CH-BUa-complete}. This result enables an additional
structural result, given in Section~\ref{SEC:RemoveRoots} about the
class of strong approximation problems $\BUA=\BBUA$, showing that the
class is unchanged even when allowing root operations as basic
operations.

Suppose we are given a basic $\ell_{\infty}-\BBUA$ problem $\Pi_a$ with circuits
over the basis $\{+,-,\ast,\max,\min\}$. Let $(I,k)$ denote an instance of $\Pi_a$ 
and put $\eps = 2^{-k}$. We may in polynomial time compute a circuit $C$ defining
a function $F\colon B_{\infty}^{n}\rightarrow \RR^n$ that is odd on the boundary $S_{\infty}^{n-1}$
such that $\Sol(I) = \{x\in B_{\infty}^{n}\mid F(x)=0\}$. We now provide a 
reduction from $\Pi_a$ to a $\CHA$-problem. In the reduction we will make use
of the "almost implies near" paradigm.

\begin{lemma}\label{AIN}
Let $F\colon B_{\infty}^n\rightarrow\RR^n$ be a continuous map. For any $\eps > 0$
there is a $\delta > 0$ such that if $||F(x)||_{\infty} \leq\delta$ then there is an
$x^{*}\in B_{\infty}^n$ such that $||x-x^{*}||_{\infty}\leq\eps$ and $F(x^{*})=0$.
\end{lemma}
\begin{proof}
Let $F$ and $\eps>0$ be given. Suppose the claim is false. Then for
any $n\in\NN$ there is an $x_n$ such that $||F(x)||_n\leq 1/n$ and if $x^{*}\in B_{\infty}^n$
has $||x_n-x^{*}||_{\infty}\leq \eps$ then $F(x^{*})\neq 0$.  By compactness the 
Bolzano-Weierstrass theorem implies the existence of a subsequence $\{x_{n_i}\}$ converging 
to some $x^{*}\in B_{\infty}^n$. By continuity of $F$ and $||\cdot||_{\infty}$ we get that
$||F(x^{*})||_{\infty}=\lim_{i\to\infty}||F(x_{n_i})||_{\infty}= 0$,
showing that $F(x^{*}) = 0$. However, for sufficiently large $i\in\NN$ it holds that
$||x_{n_i}-x^{*}||\leq\eps$ contradicting the choice of the $x_n$. 
\end{proof}

This lemma says that for any $\eps > 0$, if $||F(x)||_{\infty}$ is sufficiently close to being zero, then $x$ is $\eps$-close
to a real zero of $F$. When $F$ is computed by an algebraic circuit of polynomial size, it follows by the results in Section~\ref{SEC:ToolsRAG} that there exists some fixed polynomial $q$ with 
integer coefficients such that the above lemma holds true for some $\delta\geq (\eps)^{2^{q(|I|)}}.$
The lemma then holds true for $\delta = (\eps)^{2^{q(|I|)}}$, and we may construct
this number using a circuit of polynomial size by repeatedly squaring the number $\eps$ exactly
$q(|I|)$ times. This number will be used by the feedback agents in our $\CHA$ instance
in order to ensure that any solution gives a solution to the $\ell_{\infty} -\BBUA$ instance.

\subsection{Overview of the Reduction} 

\paragraph{Overview.}

As in previous works, we describe a consensus halving instance on an
interval $A=[0,M]$, where $M$ is bounded by a polynomial in $\abs{I}$,
rather than the interval~$[0,1]$. This instance may then be translated
to an instance on the interval~$[0,1]$ by simple
scaling. Like~\cite{EC:Filos-RatsikasHSZ20}, in the leftmost end of
the instance we place the \textit{Coordinate-Encoding} region
consisting of $n$ intervals. In a solution $S$, these intervals will
encode a value $x\in [-1,1]^n$. A \textit{circuit simulator} $C$ will
simulate the circuit of $F$ on this value $x$. The circuit simulators
will consist of a number of agents each implementing one gate of the
circuit. However, such a circuit simulator may fail in simulating $F$
properly, so we will use a polynomial number of circuit simulators
$C_1,\dots, C_{p(n)}$. Each of these circuit simulators will output
$n$ values $[C_j(x)]_1,\dots, [C_j(x)]_n$ into intervals
$I_{1j}, \dots, I_{nj}$ immediately after the simulation. Finally, we
introduce the so-called \textit{feedback agents} $f_1,\dots, f_n$. The
agent $f_i$ will have some very thin \textit{Dirac blocks} centered in
each of the intervals $I_{ij}$ where $j\in [p(n)].$ These agents will
ensure that if $z$ is an exact solution to the CH instance, then the
encoded value $x$ satisfies that $||F(x)||_{\infty}$ is sufficiently
small that we may conclude that $x$ is $\eps$-close to a zero $x^{*}$
of $F$.

\paragraph{Label Encoding.} 

For a unit interval $I$ we let $I^{\pm}$ denote the subsets of $I$ assigned the
corresponding label. We define the \textit{label encoding} of $I$ to be a value in
$[-1,1]$ given by the formula $v_{l}(I):=\lambda(I^{+})-\lambda(I^{-})$, where
$\lambda$ denotes the Lebesgue measure on the real line $\RR$. This makes sense
as $I^{\pm}$ is measurable, because they are the union of a finite number of intervals.

\paragraph{Coordinate-Encoding Region.}

The interval $[0,n]$ is called the \textit{Coordinate-Encoding} region.
For every $i\in [n],$ the subinterval $[i-1,i]$ of the Coordinate-Encoding region 
encodes a value $x_i := v_l([i-1,i])$ via the label encoding. 

\paragraph{Position Encoding.} 

For an an interval $I$ which contains only a single cut, thus dividing $I$ into two
subintervals $I=I_a\cup I_b,$ we define the \textit{position encoding} of $I$ to
be the value $v_p(I):=\lambda(I_1)-\lambda(I_2)$. We note that $v_p(I)=v_l(I)$
if the labeling sequence is $-/+$, and $v_p(I)=-v_l(I)$ in the case the labeling
sequence is $+/-$. 

\paragraph{From Label to Position.} 

Before a circuit simulator there is a \textit{sign detection} interval $I_s$
which detects the labeling sequence. Unless it contains a stray cut, this interval
will encode a sign $s=\pm 1$ (to be precise $1$ if the label is $+$ and $-1$ is
the label is $-$). By placing agents that flip the label as indicated below, we may now obtain
position encodings of the values $sx_1,\dots, sx_n$. These values will be read-in
as inputs to the subsequent circuit simulator.

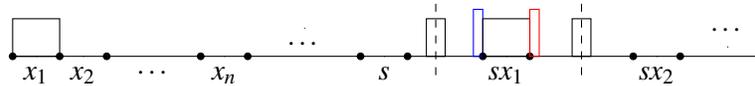
\begin{center}
\begin{tikzpicture}[scale = 2.5]
\draw (1,0) -- (5,0);
\filldraw [black] (1,0) circle (0.5pt) node[anchor = south east]{};
\filldraw [black] (1.25,0) circle (0.5pt) node[anchor = north]{};
\filldraw [black] (1.5,0) circle (0.5pt) node[anchor = north]{};
\draw (1.75,0) node[anchor = north]{$\cdots$};
\filldraw [black] (2,0) circle (0.5pt) node[anchor = north]{};
\filldraw [black] (2.25,0) circle (0.5pt) node[anchor = south]{};

\draw (1.125,0) -- (1.125,0) node[anchor = north]{$x_1$};
\draw (1.375,0) -- (1.375,0) node[anchor = north]{$x_2$};
\draw (2.125,0) -- (2.125,0) node[anchor = north]{$x_n$};

\draw (2.55,0.15) -- (2.55,0.15) node[anchor = north]{$\cdots$};

\filldraw [black] (2.85,0) circle (0.5pt) node[anchor = south]{};
\filldraw [black] (3.1,0) circle (0.5pt) node[anchor = south]{};
\draw (2.975,0) -- (2.975,0) node[anchor = north]{$s$};

\draw [black] (3.2,0) rectangle (3.3,0.2);
\draw [dashed] (3.25,-0.1) -- (3.25,0.3);

\filldraw [black] (3.5,0) circle (0.5pt) node[anchor = south]{};
\filldraw [black] (3.75,0) circle (0.5pt) node[anchor = south]{};

\draw [black] (3.975,0) rectangle (4.075,0.2);
\draw [dashed] (4.025,-0.1) -- (4.025,0.3);

\filldraw [black] (4.3,0) circle (0.5pt) node[anchor = south]{};
\filldraw [black] (4.55,0) circle (0.5pt) node[anchor = south]{};

\draw [blue] (3.45,0) rectangle (3.5,0.25);
\draw [black] (3.5,0) rectangle (3.75,0.2);
\draw [red] (3.75,0) rectangle (3.8,0.25);


\draw [black] (1,0) rectangle (1.25,0.2);

\draw (4.8,0.05) -- (4.8,0.05) node[anchor = south]{$\cdots$};

\draw (3.625,0) -- (3.625,0) node[anchor = north]{$s x_1$};
\draw (4.425,0) -- (4.425,0) node[anchor = north]{$s x_2$};
\end{tikzpicture} 
\end{center}

\paragraph{Circuit Simulators.} 

As mentioned above, the circuit simulator $C_j$ will read-in the values
$s_j x_1,\dots, s_j x_n$ and simulate the circuit computing $F$ on this
input. They then output their values into $n$ intervals immediately
after the simulation.

\paragraph{Feedback Agents.} 
By the discussion after the proof of Lemma~\ref{AIN} we may by repeated squaring construct a circuit
of polynomial size in $|I|$ computing a tiny number $\delta>0$
such that if $||F(x)||_{\infty}\leq\delta$ then $x$ is $(\eps /2)$-close to a
zero of $F$. Now fix $i\in [n]$ and let $c_{ij}$ denote the centre of the feedback
interval $I_{ij}$ outputs the value $[C_j(s_j\cdot x)]_i$. We then
define the $i$th feedback agent to have constant density $1/\delta$ in the
intervals $[c_{ij}-\delta/2,c_{ij}+\delta/2].$

The reason for having the feedback agents have these very narrow
Dirac blocks is that if $F_i(x)>\delta$ for some $i$, then in any
of the "uncorrupted" circuits (i.e. circuits outputting the correct
values) all the density of the $i$th agent will contribute to the
same label. Moreover, we will show using the boundary condition
of $F$ that the contribution is to the same label in all the
uncorrupted circuit simulators. This will contradict that the
feedback agents should value $I^{+}$ and $I^{-}$ equally. That
is the feedback agents ensure that $||F(x)||_{\infty} \leq\delta$ if $x$ is the value encoded by
an exact solution to the consensus halving instance we construct.

\paragraph{Stray Cuts.} Any of the agents implementing one of the
gates in a circuit simulator will force a cut to be placed in an
interval in that same circuit simulator. The only agents whose cuts
we have no control over are the $n$ feedback agents.The expectation
is that these agents should make cuts in the Coordinate-Encoding region
that flip the label.
If they do not do this we will call it a \textit{stray cut}. If a circuit
simulator contains a stray cut, we will say nothing about its value. 

\begin{observation}
If it is not the case that every unit interval encoding a coordinate $x_i$ in the Coordinate-Encoding
region contains a cut that flips the label, then the encoded point $x\in B_{\infty}^n$ will lie on the
boundary $S_{\infty}^n$.  With this in mind we may ensure that $x\in S_{\infty}^n$ or
$s_1=s_2=\cdots =s_{p(n)}=\pm 1$ where the sign is the same as the label of the first interval. 
This can be done by, if necessary, placing one single-block agent after the Coordinate-Encoding region and each
of the circuit simulators (if placing such an agent is necessary depends on, respectively, the number of variables
$n$ and the size of the circuits). 
\label{OBS:SignFlips}
\end{observation}

\subsection{Construction of Gates}
In this section we describe how to construct Consensus-Halving agents implementing
the required gates $\{+,-,\ast,\max,\min\}$. First, we show that we may transform the
circuit such that all gates only take values in the interval $[-1,1]$ on input
from $B_{\infty}^n$. 

\paragraph{Transforming the Circuit.}

By propagating every gate to the top of the circuit we may assume that
the circuit is layered. Let $C'$ denote the resulting circuit. By repeated squaring we may  
maintain a gate with value $1/2^{2^d}$ in the $d$th layer. Suppose $g=\alpha(g_1,g_2)$
is a gate with inputs $g_1,g_2$ in layer $d$. We modify the gates as follows: if $\alpha\in
\{+,-,\max,\min\}$ then we multiply $g_i$ by $1/2^{2^d}$ before applying $\alpha$;
if $\alpha =\ast,$ then we multiply the input by $1$ before applying $\alpha.$ Finally,
we transform $C'$ into the circuit $C''$ as follows: on input $x$, the circuit $C''$
multiplies the input by $1/2$ and then evaluates $C'$ on input $x/2$. Inductively, one may
show that if $g$ is a gate in layer $d$ in the circuit $C'$, then the corresponding gate in
in the circuit $C''$ has value $g/2^{2^d}$. As all the gates are among $\{+,-,\ast,\max,
\min\}$, this ensures that all the gates in $C''$ take values in $[-1,1]$. 

\paragraph{Addition Gate [$G_{+}].$} 

We may construct an addition gate using two agents. The first agent has two unit input intervals that
we assume contain one cut each. This then forces a cut in the long output interval that has length 3. 
The second agent then truncates this value.

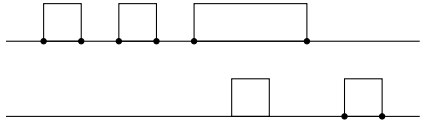
\begin{center}
\begin{tikzpicture}[scale = 0.5]
\draw (0,0) -- (11,0);
\draw (1,0) rectangle (2,1);
\filldraw [black] (1,0) circle (2pt) node[anchor = south]{};
\filldraw [black] (2,0) circle (2pt) node[anchor = south]{};

\draw (3,0) rectangle (4,1);
\filldraw [black] (3,0) circle (2pt) node[anchor = south]{};
\filldraw [black] (4,0) circle (2pt) node[anchor = south]{};

\draw (5,0) rectangle (8,1);
\filldraw [black] (5,0) circle (2pt) node[anchor = south]{};
\filldraw [black] (8,0) circle (2pt) node[anchor = south]{};

\draw (0,-2) -- (11,-2);

\draw (6,-2) rectangle (7,-1);
\draw (9,-2) rectangle (10,-1);
\filldraw [black] (9,-2) circle (2pt) node[anchor = south]{};
\filldraw [black] (10,-2) circle (2pt) node[anchor = south]{};
\end{tikzpicture} 
\end{center}

\paragraph{Constant Gate [$G_{\zeta}].$} Let $\zeta\in [-1,1]\cap\QQ$ be a rational
constant. The agent will have a block of unit height in the sign interval and a block of
width $\zeta/2$ and height $2/\zeta$ centered in another interval.

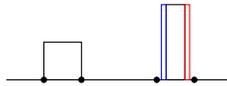
\begin{center}
\begin{tikzpicture}[scale = 0.5]
\draw (0,0) -- (6,0);
\draw (1,0) rectangle (2,1);
\filldraw [black] (1,0) circle (2pt) node[anchor = south]{};
\filldraw [black] (2,0) circle (2pt) node[anchor = south]{};

\filldraw [black] (4,0) circle (2pt) node[anchor = south]{};
\filldraw [black] (5,0) circle (2pt) node[anchor = south]{};
\draw (4.25,0) rectangle (4.75,2);
\draw [blue] (4.125,0) rectangle (4.25,2);
\draw [red] (4.75,0) rectangle (4.875,2);

\end{tikzpicture} 
\end{center}


Before proceeding with the remaining gates, we construct a general function gate,
an agent that implements any decreasing function. 

\paragraph{Function Gate [$G_h$].} Let $-1\leq a<b\leq 1$ and $-1\leq c<d\leq 1$
be rational numbers and consider a continuously differentiable map $h\colon
[a,b]\rightarrow [c,d]$ satisfying $h(a)=d$ and $h(c) = c$. Let $\overline{h}$
denote the extension of $h$ that is constant on $[-1,a]$ and $[b,1]$. We now construct 
an agent with input interval $I$ and output interval $O$ computing this map, that is the
agent should force a cut in the output interval such that $\overline{h}(v_p(I)) = v_p(O).$ 

The agent that we construct has a block of height $2/(d-c)$ in the sub-interval
$[(c+1)/2,(d+1)/2]$ of the output interval and density $f(z):=-2h'(2z-1)/(d-c)$ in the
sub-interval $((a+1)/2,(b+1)/2)$ of the input interval. We note that $f$ is positive in this
interval as $h$ is assumed to be a decreasing map, so it makes sense for the agent 
to have density $f.$ One may verify that the agent values the input interval and output
interval equally. We further add two rectangles to the output interval colored blue and
red in the sketch below. These will ensure that if the cut in the input interval is placed
at $z\leq (a+1)/2$ such that $v_p(I)\leq a$, then the cut in the output interval must be
placed at $z^{*}=(d+1)/2$, meaning that $v_p(O)=d.$ Similarly, if $v_p(I)\geq b$ then
$v_p(O)=c$.

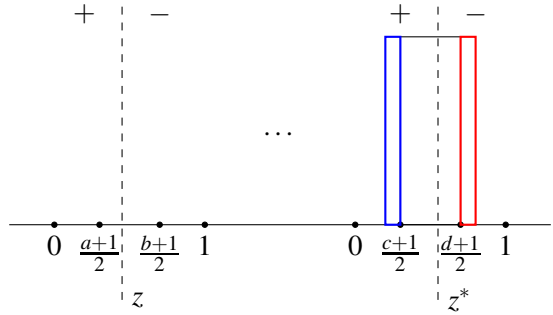
\begin{center}
\begin{tikzpicture}[scale = 2]
\draw (0.7,0) -- (4.3,0);
\filldraw [black] (1,0) circle (0.5pt) node[anchor = north]{0};
\filldraw [black] (2,0) circle (0.5pt) node[anchor = north]{1};

\filldraw [black] (1.3,0) circle (0.5pt) node[anchor = north]{$\tfrac{a+1}{2}$};
\filldraw [black] (1.7,0) circle (0.5pt) node[anchor = north]{$\tfrac{b+1}{2}$};

\filldraw [black] (3,0) circle (0.5pt) node[anchor = north]{0};
\filldraw [black] (4,0) circle (0.5pt) node[anchor = north]{1};
\filldraw [black] (3.3,0) circle (0.5pt) node[anchor = north]{$\tfrac{c+1}{2}$};
\filldraw [black] (3.7,0) circle (0.5pt) node[anchor = north]{$\tfrac{d+1}{2}$};

\draw (3.3,0) rectangle (3.7,5/4);
\draw[thick,blue] (3.2,0) rectangle (3.3,1.25);
\draw[thick,red] (3.7,0) rectangle (3.8,1.25);

\draw[thin, dashed] (1.45,-0.5) node[anchor = west]{$z$}  -- (1.45,1.5);
\draw[thin, dashed] (3.55,-0.5) node[anchor = west]{$z^{*}$}  -- (3.55,1.5);

\draw (1.2,1.4) node{$+$};
\draw (1.7,1.4) node{$-$};

\draw (3.3,1.4) node{$+$};
\draw (3.8,1.4) node{$-$};

\draw (2.5,0.6) node{$\cdots$};
\end{tikzpicture} 
\end{center}
Suppose cuts are placed in $z$ in the input interval and in $z^{*}$ in the output
interval. As the agent must value the parts with positive and negative label equally,
we get the equality
\begin{align*}
1 = \int _{(a+1)/2}^{z}\tfrac{-2h'(2t-1)}{d-c}\, dt + \big(z^{*}-\tfrac{c+1}{2}\big) \tfrac{2}{d-c}
\end{align*}
From this we obtain that
\begin{align*}
d-c &=-\int_{a}^{2z-1}h'(u)\, du+2z^{*}-c-1\\
&=-h(2z-1)+d+2z^{*}-c-1
\end{align*}
where we use that $h(a)=d$ by assumption. We conclude that $h(2z-1)=2z^{*}-1$,
that is we obtain the equality $h(v_p(I))=v_p(O).$  


Using this general function gate, we may now build up the remaining
gates required by the circuit.

\paragraph{Multiplication By -1 Gate $[G_{-(\cdot)}]$.}  

In order to realise this gate, we consider the function $h\colon [-1,1]\rightarrow [-1,1]$
given by $x\mapsto -x$. The agent's density function in the input interval is then
given by $f(z)=1$.
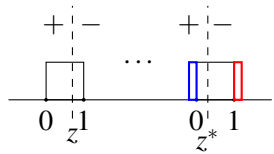
\begin{center}
\begin{tikzpicture}[scale = 0.5]
\draw (0,0) -- (7,0);

\filldraw [black] (1,0) circle (1pt) node[anchor = north]{0};
\filldraw [black] (2,0) circle (1pt) node[anchor = north]{1};

\filldraw [black] (5,0) circle (1pt) node[anchor = north]{0};
\filldraw [black] (6,0) circle (1pt) node[anchor = north]{1};

\draw (1,0) rectangle (2,1);
\draw (5,0) rectangle (6,1);
\draw[thick, blue] (4.8,0) rectangle (5,1);
\draw[thick, red] (6,0) rectangle (6.2,1);

\draw[dashed] (1.7,-0.5) node[anchor=north] {$z$} -- (1.7,2.5);
\draw[dashed] (5.3,-0.5) node[anchor=north] {$z^{*}$} -- (5.3,2.5);

\draw (1.2,2) node{$+$};
\draw (2.2,2) node{$-$};

\draw (3.5,1) node{$\cdots$};

\draw (4.8,2) node{$+$};
\draw (5.7,2) node{$-$};
\end{tikzpicture} 
\end{center}

\paragraph{Subtraction Gate [$G_{-}$].}

We may build this using the gates $G_{-(\cdot)}$ and $G_{+}$. 

\paragraph{Multiplication by $\zeta\in [-1,1]$ [$G_{\cdot\zeta}$].}

If $\zeta <0$ we mahy construct $G_{\cdot\zeta}$ as a function gate
using the function $h\colon [-1,1]\rightarrow [\zeta,-\zeta]$. If $\zeta >0$
we construct using $-\zeta$ and a minus gate, i.e. $G_{\cdot\zeta} = -G_{\cdot (-\zeta)}$.

\paragraph{Maximum Gate [$G_{\max}]$.} 

First we show how to construct a gate computing the absolute value
of the input. We may construct gates $G_1,G_2$ such that $G_1(x)=
-\max(x,0)$ and $G_2(x)=\max(-x,0)$ as function gates by using the 
functions $h_1\colon [0,1]\rightarrow [-1,0]$ given by $x\mapsto -x$
and $h_2\colon [-1,0]\rightarrow [0,1]$ given by $x\mapsto -x$. Now,
we may constrcuct the absolute value gate as $G_{|\cdot|} = -G_1+G_2$.
We may now construct $G_{\max}$ by using the formula $\max(x,y)
=(x+y+|x-y|)/2$.

\paragraph{Minimum Gate [$G_{\min}]$.}

We may build this using $\min(x,y)=x+y-\max(x,y).$


\paragraph{Multiplication Gate [$G_{*}$].} 

We start off by constructing a gate squaring the input. First we construct
$G_1$ and $G_2$ as function gates with respect to $h_1\colon [-1,0]
\rightarrow [0,1]$ given by $x\mapsto x^2$ and $h_2\colon [0,1]
\rightarrow [-1,0]$ given by $x\mapsto -x^2$. Then we may construct
the squaring gate as $G_{(\cdot)^2}=G_1-G_2$. 
Now we may use the previously constructed gates to make a multiplication
gate via the identity $xy=((x+y)^2-x^2-y^2)/2$.

\subsection{Describing valuation functions as circuits.} 

In the description above, we described the valuations
of the agents by providing formulas for their densities. However, an instance of $\CH$
actually consists of a list of algebraic circuits computing the distribution functions of the agents.
In order to construct gates, it is
sufficient for agents to have densities that are piece-wise polynomial. Therefore, consider
an agent with polynomial densities $f_i$ in the intervals $[a_i,b_i)$ for $i = 1,\dots, s$,
and let $F_i$ denote the indefinite integral of $f_i$. We note that $F_i$ is a polynomial
so it may be computed by an algebraic circuit.  Now we claim that the distribution function
of this agent may be computed by an algebraic circuit via the formula
\begin{align}
F(x) = \sum_{i = 1}^{s}[F_i(\max(a_i,\min(x,b_i)))-F_i(a_i)]
\end{align}
This is the case, because the summands will be equal to $F_i(a_i)-F_i(a_i) = 0$ if $x<a_i,$
to $F_i(x)-F_i(a)$ if $a_i\leq x\leq b_i$ and to $F_i(b)-F_i(a)$ if $x>b_i$, meaning that this
formula does indeed calculate the valuation of the agent in the interval $[0,x].$

\subsection{Reduction and Correctness}

Recall that we are given an instance $(F,\eps)$ of the $\BBUA$ problem and
that we have to construct an instance of the $\CHA$ problem. The reduction now
outputs an instance of the $\CHA$ problem where the consensus halving 
instance is constructed as above with $p(n)=2n+1$ circuit simulators and 
the approximation parameter is given by $\eps ' = \eps/(4n)$. Let $z$ denote a solution
to this $\CHA$ instance. By definition, there exists an exact solution $z^{*}$ to the
consensus-halving problem such that $\norm{z-z^{*}}_{\infty}\leq\eps '$. 

Let $x$ and $x^{*}$ denote the values encoded by respectively $z$ and $z^{*}$
in the Coordinate-Encoding region. Suppose, generally, we are given an interval $I$
with a number of cut points $t_1,\dots, t_s$. Moving a cut point by a distance $\leq\eps'$ 
we create a new interval $I'$. This changes the label encoding by at most $2\eps'$, that is
$|v_l(I)-v_l(I')|\leq 2\eps'$. Succesively, if we move all the cuts by a distance
$\leq\eps'$, then we get an interval $I^{*}$ such that $|v_l(I)-v_l(I^{*})|
\leq 2s\eps' $. As $\norm{z-z^{*}}_{\infty}\leq\eps'$ and any of the subintervals 
in the Coordinate-encoding region can contain at most $n$ cuts, we conclude
that $\norm{x-x^{*}}_{\infty}\leq 2n\eps '=2n(\eps /(4n))=\eps /2$. In order to show that
$x$ is $\eps$-close to a zero of $F$, it now suffices by the triangle inequality to show that
$x^{*}$ is $(\eps/2)$-close to a zero of $F$. This will follow from the two following
lemmas.

\begin{lemma}
If there are no stray cuts in the exact solution $z^{*}$, then the associated value
$x^{*}$ encoded in the Coordinate-encoding region satisfies $F(x^{*})=0$. 
\end{lemma}
\begin{proof}
We recall that if the solution $z^{*}$ contain no stray cuts, then the signs of all
the circuit simulators are equal $s_1=\dots = s_{2n+1}=s$ where $s=\pm 1$. Furthermore, all the
circuit simulators will output the same values $F_1(s x^{*}),\dots, F_n(s x^{*})$
into the feedback intervals. Thus, there can be no cancellation, so in order for the
feedback agents to value the positive and negative part equally it must be the 
case that $F(s x^{*})=0$. 
\end{proof}

\begin{lemma}
If there is a stray  cut in the exact solution $z^{*}$, then the associated value
$x^{*}$ encoded in the Encoding-region satisfies the inequality 
$||F(x^{*})||_{\infty}\leq\delta$. 
\end{lemma}
\begin{proof}
Suppose toward contradiction that $|F(x)_i|>\delta$ for some $i$. Without loss
of generality we assume that $F(x)_i>\delta$.  As there is a stray cut, the Coordinate-Encoding
region can contain at most $n-1$ cuts. Thus, at least one of the coordinates
$x^{*}_i$ must be $\pm 1$ showing that $x^{*}\in S^{n-1}$. From this and the 
boundary condition we conclude that $F(x^{*})=-F(-x^{*})$. Furthermore, there is 
at most $n$ stray cuts, so at most $n$ circuit simulators can become corrupted. This
means that $n+1$ circuit simulators work correctly. Now suppose that the circuit
simulator $C_j$ is uncorrupted. If the label is $s_j = +1$, then $C_j$ will output $F(x)$ into
the feedback region and the labeling sequence will be $+/-$; if the label is $s_j = -1$
then $C_j$ will output $F(-x)=-F(x)$ into the feedback region and the labeling sequence
will be $-/+$. This is indicated below:
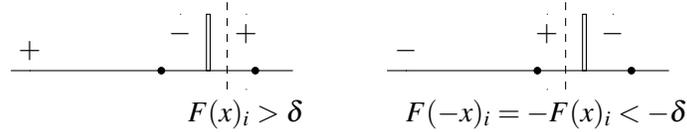
\begin{center}
\begin{tikzpicture}[scale = 2.5]
\draw (1,0) -- (2.5,0);

\draw (3,0) -- (4.5,0);

\draw (1.1,0) -- (1.1,0) node[anchor = south]{$+$};
\draw (3.1,0) -- (3.1,0) node[anchor = south]{$-$};

\filldraw [black] (1.8,0) circle (0.5pt) node[anchor = south east]{};
\filldraw [black] (2.3,0) circle (0.5pt) node[anchor = south east]{};

\filldraw [black] (3.8,0) circle (0.5pt) node[anchor = south east]{};
\filldraw [black] (4.3,0) circle (0.5pt) node[anchor = south east]{};

\draw (2.04,0) rectangle (2.06,0.3);
\draw (4.04,0) rectangle (4.06,0.3);
\draw (1.9,0.3) -- (1.9,0.3) node[anchor = north]{$-$};
\draw (2.25,0.3) -- (2.25,0.3) node[anchor = north]{$+$};
\draw (2.25,-0.1) -- (2.25,-0.1) node[anchor = north]{$F(x)_i>\delta$};

\draw [dashed] (2.15,-0.1) -- (2.15, 0.4);
\draw [dashed] (3.95,-0.1) -- (3.95, 0.4);
\draw (3.85,0.3) -- (3.85,0.3) node[anchor = north]{$+$};
\draw (4.2,0.3) -- (4.2,0.3) node[anchor = north]{$-$};
\draw (3.85,-0.1) -- (3.85,-0.1) node[anchor = north]{$F(-x)_i=-F(x)_i<-\delta$};
\end{tikzpicture} 
\end{center}
From this we conclude that the $n+1$ uncorrupted circuit simulators
altogether contribute $(n+1)\delta$ to the part with negative label. However,
the $n$ corrupted circuit simulators can contribute at most $n\delta$ to the 
part with positive label. This implies that $f_i$ cannot value the negative
and positive part equally. This contradicts the assumption that $z^{*}$ is an 
exact consensus-halving. We conclude that $||F(x^{*})||_{\infty}\leq \delta$.
\end{proof} 

By the two lemmas above, it follows that the value $x^{*}$ encoded by the exact
consensus-halving $z^{*}$ satisfies the inequality $||F(x^{*})||_{\infty}\leq\delta$. 
By choice of $\delta$, this implies that there exists some $x^{**}$ such that
$\norm{x^{*}-x^{**}}_{\infty}\leq\eps /2$ and $F(x^{**})=0$. From the discussion before
the two lemmas, it follows that $x$ is $\eps$-close to a zero of $F$ and is thus a
solution to the $\BBUA$ instance $(F,\eps)$. 

\paragraph{Mapping back a Solution.}  

What remains is to show that we may recover a solution $x$ to the $\BBUA$ instance
from the solution $z$ to the $\CHA$ instance. Recall that in a solution $z = (z_1,\dots, z_N)$
to the consensus-halving problem $|z_i|$ and $\sgn(z_i)$ represents the length and label
of the $i$th interval. For $i\leq n$ and $j\leq n+1$ we introduce
\begin{align*}
& t_j = \sum_{k=1}^{j-1}|z_k|\\
& x_{ij}^{+} = \max(0,\min(t_{j-1}+z_j,i)-\max(t_{j-1},i-1))\\
& x_{ij}^{-} = \max(0,\min(t_{j-1}-z_j,i)-\max(t_{j-1},i-1))
\end{align*}
These numbers may be computed efficiently by a circuit over $\{+,-,\max,\min\}$. We notice
that if $z_j>0$ then $x_{ij}^{-}=0$ (and if $z_j<0$ then $x_{ij}^{+}=0$). Furthermore, by 
checking a couple of cases, one finds that if $z_j>0$ (respectively $z_j <0$) then $x_{ij}^{+}$
(respectively $x_{ij}^{-}$) is the length of the $j$th interval that is contained in $[i-1,i]$. As the
coordinate-encoding region can contain at most $n$ cuts (corresponding to at most $n+1$ intervals),
we deduce from the above that the values encoded can be computed as
\begin{align*}
x_i = \sum_{j = 1}^{n+1}x_{ij}^{+}-x_{ij}^{-}
\end{align*}
for every $i\leq n$. If there is a stray cut then both $x$ and $-x$ are valid solutions by the
boundary condition of $F$. If there is no stray cut, then $s_1=s_2=\cdots = s_{p(n)}=s=\sgn(z_1)$
by Observation~\ref{OBS:SignFlips} and in this case we may recover a solution as $sx$.

\subsection{Removing Root Gates.}
\label{SEC:RemoveRoots}

In this subsection, we argue by going through $\CHA$ that the strong approximation problems 
$\BUA = \BBUA$ do not change even if we allow the circuits to use root-operations as basic
operations. 

\begin{proposition}
The class $\ell_{\infty}-\BBUA$ remains unchanged even if we allow the circuits to use root-gates.  
\end{proposition}
\begin{proof}
Let $\Pi_a$ be a basic $\ell_{\infty}-\BBUA$ problem where the circuits are allowed to use
gates from the basis $\{+,-,\ast,\max,\min,\sqrt[k]{}\}$. In the previous section, we constructed a
polynomial time reduction 
from $\Pi_a$ to a $\CHA$ problem $\Gamma_a$ in such a way that the circuits computing the 
distribution functions of the agents are defined over $\{+,-,\ast,\max,\min\}$. Namely, 
the root gates can be implemented by first noting that the power-gate $(\cdot)^k$ can be 
implemented by an agent with polynomial densities by using the general function gate 
construction. Then, in order to construct an agent implementing the root gate we simply 
interchange the input interval and output interval of the power-gate. By the proof of 
the result of Deligkas et al. that $\CH$ is contained in $\BU$, the problem $\Gamma_a$ 
polynomial time reduces to a $\ell_1 -\BUA$ problem $\Lambda$ that only uses gates
from $\{+,-,\ast,\max,\min\}$. By Proposition~\ref{PROP:BasicBUaToBasicBBUa}, $\Lambda$ reduces
to a basic $\ell_1 -\BBUA$ problem $\Xi$ which again uses only gates from $\{+,-,\ast,\max,\min\}$. 
Finally, by Proposition~\ref{PROP:BBU-p-to-inf}, $\Xi$ reduces to a basic $\ell_\infty -\BBUA$ problem, again
using only gates from $\{+,-,\ast,\max,\min\}$. Altogether, we see that $\Pi_a$ polynomial time reduces to
a $\ell_{\infty}-\BBUA$ without root-gates.
\end{proof}

\printbibliography

\end{document}